\documentclass[sigconf,nonacm]{acmart}
\AtBeginDocument{%
  \providecommand\BibTeX{{%
    Bib\TeX}}}

\setcopyright{acmlicensed}
\setcctype{by}
\copyrightyear{2027}
\acmYear{2027}
\acmDOI{XXXXXXX.XXXXXXX}
\acmConference[WWW '27]{the ACM Web Conference 2027}{May 10--14, 2027}{Dublin, Ireland}
\acmISBN{978-1-4503-XXXX-X/2018/06}

\acmSubmissionID{3281}

\usepackage{tcolorbox}
\usepackage{xspace}

\usepackage{amsthm}
\usepackage{amsmath} 

\theoremstyle{definition}
\theoremstyle{plain}
\newtheorem{definition}{Definition}
\newtheorem{problem}{Problem}
\newtheorem{lemma}{Lemma}

\usepackage{algorithm}
\usepackage{algpseudocode}

\usepackage{booktabs}
\usepackage{multirow}

\usepackage{graphicx}

\usepackage{enumitem}

\usepackage{url}

\usepackage{hyperref}

\usepackage{xcolor}

\usepackage{enumitem}

\newcommand{\spara}[1]{\smallbreak\noindent{\bf{#1}}}

\newcommand{\squishlist}{\begin{list}{$\bullet$}
  { \setlength{\itemsep}{0pt}
     \setlength{\parsep}{3pt}
     \setlength{\topsep}{3pt}
     \setlength{\partopsep}{0pt}
     \setlength{\leftmargin}{1.5em}
     \setlength{\labelwidth}{1em}
     \setlength{\labelsep}{0.5em} } }
\newcommand{\squishend}{
  \end{list}  }

\usepackage{amsmath, amsthm, mathtools}

\newtheorem{theorem}{Theorem}
\newtheorem{assumption}{Assumption}

\usepackage{subcaption}

\usepackage{stfloats}

\usetikzlibrary{arrows.meta}
\hypersetup{colorlinks=true,linkcolor=blue,citecolor=blue,urlcolor=blue}

\usepackage{multirow}
\usepackage{booktabs}

\begin{document}



\title[Not Everyone Has a Voice, or an Ear, on Social Media: The Poster-Observer Model of Opinion Dynamics]{Not Everyone Has a Voice, or an Ear, on Social Media:\\ The Poster-Observer Model of Opinion Dynamics}

\author{Nikos Theologis}
\affiliation{%
  \institution{University of Ioannina}
  \city{Ioannina}    
  \country{Greece}
}
\affiliation{%
  \institution{Archimedes/Athena Research Center}
  \city{Athens}    
  \country{Greece}
}
\email{n.theologis@athenarc.gr}

\author{Aristides Gionis}
\affiliation{%
  \institution{KTH Royal Institute of Technology}
  \city{Stockholm}  
  \country{Sweden}
}
\email{argioni@kth.se}

\author{Evaggelia Pitoura}
\affiliation{%
  \institution{University of Ioannina}
  \city{Ioannina}    
  \country{Greece}
}
\affiliation{%
  \institution{Archimedes/Athena Research Center}
  \city{Athens}    
  \country{Greece}
}
\email{pitoura@uoi.gr}

\author{Evimaria Terzi}
\affiliation{%
  \institution{Boston University}
  \city{Boston}  
  \country{USA}
}
\email{evimaria@bu.edu}

\author{Panayiotis Tsaparas}
\affiliation{%
  \institution{University of Ioannina}
  \city{Ioannina}    
  \country{Greece}
}
\affiliation{%
  \institution{Archimedes/Athena Research Center}
  \city{Athens}    
  \country{Greece}
}
\email{tsap@uoi.gr}

\begin{abstract}
Classic modeling of opinion dynamics via the Friedkin-Johnsen (FJ) model relies on the unrealistic assumption of synchronous communication among network users. 
Social media, however, exhibit extreme participation asymmetries: a vast majority of users silently observe, while a tiny minority creates and broadcasts most of the content. 
To capture this reality, we introduce the Poster-Observer FJ (PO-FJ) model, which decouples user communication into distinct posting and observing behaviors.
This formulation fundamentally alters network dynamics by preventing convergence to a static equilibrium, driving opinions into perpetual oscillations. 
We establish closed-form expressions for these new dynamics and generalize classic metrics of public opinion, polarization, and disagreement under the PO-FJ model. 
By leveraging the novel Poster-Observer dynamics, we introduce a realistic and scalable intervention algorithm for conflict mitigation. 
Through experiments on synthetic and real-world networks, we demonstrate the efficacy and efficiency of our approach in reducing social-network conflict.
\end{abstract}





\maketitle

\section{Introduction}\label{sec:intro}

Social media platforms have become the primary arena where modern public discourse unfolds, news is consumed, and societal opinions are shaped \cite{lazer2009computational, flaxman2016filter}. 
Every day, billions of users interact throughout
the Web, exchanging ideas, expressing viewpoints, and influencing the beliefs of their peers. 
Because these interactions directly affect real-world phenomena, ranging from political polarization and echo chambers to consensus formation \cite{adamic2005political, conover2011political, del2016spreading}, mathematical models of opinion dynamics \cite{social-dynamics-survey, das2014modeling} have become central to the analysis of online social networks.

Among these frameworks, the Friedkin-Johnsen (FJ) model \cite{Friedkin_Johnsen} is one of the most established and widely studied models of opinion formation. 
Under FJ, individuals update their expressed opinions as a convex combination of their own inner opinion and the expressed opinions of their connections, weighted by their stubbornness.
This captures the tension between the inner beliefs of users and their tendency to conform to the opinions of their social circle.
The FJ model has been used extensively to study social-network polarization and disagreement \cite{Howbad,matakos2017measuring, Musco, Tsourakakis, NEURIPS2021}. 
Crucially, FJ relies on the unrealistic assumption that communication takes place in a \emph{synchronous} manner, where all users simultaneously broadcast their opinions and read the opinions of their neighbors in the social graph.

Web platforms such as social networks, review sites, and online forums, however,
exhibit extreme participation asymmetries \cite{Nonnecke2000, nielsen2006participation, vanMierlo2014, CarronArthur2014}. 
As captured by the ``1\% rule'' and the ``90-9-1 rule'' of online communities \cite{nielsen2006participation, vanMierlo2014}, the vast majority of users act as silent \emph{lurkers} who browse their feeds rarely posting, a small fraction contribute intermittently, and a tiny minority of hyper-active \emph{broadcasters} generate most of the content. 
To capture this reality, we expand the classic FJ framework by introducing the \emph{Poster-Observer Friedkin-Johnsen} (PO-FJ) model, which decouples user communication into two distinct behaviors: \emph{posting} an opinion and \emph{observing} the opinions of others. 
Under PO-FJ, the network changes dynamically at each iteration based on who posts and who observes. 
When a user does not observe their feed, their outward edges are temporarily removed; conversely, when a user does not post, all inward edges from followers who would have read their content become inactive.
This mechanism captures diverse user behavior, spanning from passive lurkers to influential broadcasters.

We show that by introducing stochastic posting and observing user behaviors, the opinion dynamics of the network fundamentally change. 
Unlike the classic FJ model, where opinions deterministically converge to a static equilibrium, PO-FJ dynamics drive user opinions into perpetual oscillations, which converge to a \emph{stationary probability distribution}. 
As a result, the long-run state of the social network is no longer described by a single static vector, but by the \emph{expectations}, \emph{variances}, and pairwise \emph{correlations} of the opinions, for which we derive closed-form expressions. 
Furthermore, because the underlying dynamics have changed, classic measures of public opinion, polarization, and disagreement \cite{Howbad, matakos2017measuring, Musco} are computed in expectation and now depend on the opinion variance.

Following the literature on mitigating polarization and disagreement in the FJ model \cite{matakos2017measuring, Musco, matakos2020tell, NEURIPS2021}, we investigate how the new levers of the PO-FJ model can be used to reduce network conflict.
Whereas traditional interventions rely on modifying graph edges or inner beliefs, we study conflict mediation through the posting and observing behaviors.
We formalize this task as the \emph{Conflict Limitation via Exposure and Amplification Regulation} (CLEAR) problem. 
To solve it efficiently, we introduce \text{CLEAR-PGD}, a Projected Gradient Descent algorithm that allows for scaling and outperforms heuristic baselines. 
Experiments on synthetic and large-scale real-world Web networks demonstrate the efficacy and efficiency of our approach in reducing social-network conflict.

In summary, we make the following contributions:
\squishlist
    \item We introduce the Poster-Observer Friedkin-Johnsen (PO-FJ) model, which extends classic FJ by decoupling user communication into distinct posting and observing behaviors.
    \item We establish the new opinion dynamics and generalize the metrics of polarization and disagreement under PO-FJ.
    \item We formulate the CLEAR problem of minimizing network conflict by regulating user amplification and exposure.
    \item We propose \text{CLEAR-PGD}, a scalable algorithm that effectively reduces network conflict, evaluated across real-world networks.
\squishend

The rest of the paper is structured as follows. Section~\ref{sec:related_work} reviews related work. 
Section~\ref{sec:model} introduces the PO-FJ model. 
Section~\ref{sec:metrics} defines polarization and disagreement alongside the CLEAR problem and our \text{CLEAR-PGD} algorithm. 
Section~\ref{sec:experiments} presents our experimental evaluation, and Section~\ref{sec:conclusion} concludes the paper.

\section{Related Work}
\label{sec:related_work}

To the best of our knowledge, we are the first to introduce the Poster-Observer Friedkin-Johnsen (PO-FJ) model and utilize it to mitigate polarization and disagreement. 
However, our research is related to existing work on \emph{Time-Varying Network Models} and \emph{Polarization and Disagreement Mitigation} in opinion-formation dynamics.

\spara{Time-Varying Opinion Dynamics.}
Time-varying interactions have been explored across several opinion-formation frameworks. 
The Time-Varying FJ (TV-FJ) model \cite{TVFJ} generalizes classic FJ \cite{Friedkin_Johnsen} and introduces iteration-dependent interactions and stubbornness values. 
Convergence of TV-FJ is not guaranteed and conditions for stability have been studied under homophily \cite{TVFJ_Homophily} and graph structure \cite{TVFJ_Graph}.
In contrast, PO-FJ introduces an explicit interaction mechanism that allows for closed-form results. 
Work on complex networks has also examined activity-driven graphs \cite{complexFirst} and Bounded-Confidence Models (BCMs) \cite{Bounded1,Bounded2}. 
Specifically, prior literature incorporates asynchronous updates \cite{alizadeh2015activation}, random-time interactions \cite{chu2024bounded}, adaptive interaction probabilities \cite{thompsky2026bounded}, and heterogeneous node-activity levels \cite{li2023bounded} into BCMs.
Unlike our PO-FJ results, these works rely on empirical evaluation due to the non-linearity of the BCM dynamics.
Finally, asynchronous information exchange has also been modeled via the Gossip model \cite{Gossip1,Gossip2}, which introduces randomized pairwise edge interactions to FJ. 
Our PO-FJ model instead captures one-to-many broadcast and exposure behaviors.

\spara{Polarization and Disagreement Mitigation.}
A substantial body of research addresses algorithmic interventions in the FJ model to optimize social objectives, as surveyed by \citet{miyauchi2026survey}.
One primary direction focuses on structural interventions, such as rewiring network edges, to minimize polarization or disagreement \cite{Howbad, Musco, NEURIPS2021, kuhne2025optimizing}.
Another direction explores node-level interventions, modifying inner opinions \cite{gionis,matakos2017measuring, matakos2020tell} and stubbornness \cite{Tsourakakis,  Xu2022, theologis2026fairness} to reduce polarization and achieve network fairness.
How polarization can emerge in FJ with minor edge modifications has also been shown by \citet{Filter_Bubbles_FJ}.
Crucially, all of these algorithmic interventions rely on classic FJ parameters under synchronous communication. 
Under PO-FJ, we introduce new intervention mechanisms via adjustments in node exposure and~amplification.

\section{The Poster-Observer FJ Model}
\label{sec:model}

In this section, we introduce the Poster-Observer FJ model. 
We denote the identity matrix as $I$, the column vector full of ones as $\mathbf{1}$, and as $\circ$ the element-wise (Hadamard) product.

\subsection{The Friedkin-Johnsen Model}

Before presenting the PO-FJ model, we present the classic Friedkin-Johnsen (FJ) model \cite{Friedkin_Johnsen} upon which our framework builds.
Assume a directed graph $G$ consisting of $n$ nodes and $m$ edges. 
Each node $i$ has an \emph{inner opinion} $s_i \in [0,1]$ and an \emph{expressed opinion}~$z_i$. 
The inner opinions are fixed, while the expressed opinions  evolve through the iterative opinion-formation process: 
\begin{equation}
\label{eq:FJ_update}
z^{(t+1)}_i = a_i s_i + (1-a_i)\sum^n_{j=1} W_{ij} z_j^{(t)}. 
\end{equation}

In the above equation, $W_{ij}$ quantifies the importance node~$i$ places on the expressed opinion of node~$j$. 
The matrix $W \in [0,1]^{n \times n}$ is \emph{row-stochastic}, with $W_{ij}>0$ if there exists an edge from node~$i$ to node~$j$.
We call $W$ the \emph{base interaction matrix}.
Finally, $a_i\in [0,1]$ corresponds to the \emph{stubbornness} of node $i$ and quantifies its tendency to adhere to its own inner opinion.

Let $\mathbf{a}$, $\mathbf{s}$, and $\mathbf{z}$ denote the $n$-dimensional vectors of the stubbornness values, the inner, and the expressed opinions of the nodes. We can write the update Equation \eqref{eq:FJ_update} for the FJ model as
\begin{equation}
\label{eq:FJ_update_vec}
\mathbf{z}^{(t+1)} = A\mathbf{s} + (I-A)W\mathbf{z}^{(t)}, 
\end{equation}
where $A = \text{diag}(\mathbf{a})$ is the diagonal matrix of the stubbornness vector. 
It has been shown 
that the FJ opinion-formation process~converges.

\begin{theorem}[FJ Convergence~\cite{Friedkin_Johnsen}]\label{thm:fj_convergence}
If matrix $W$ is \emph{irreducible} (the underlying graph $G$ is strongly connected and aperiodic), and at least one node has stubbornness $a_i > 0$, then the sequence of opinions $\mathbf{z}^{(t)}$ under the FJ model converges to a static equilibrium
\begin{equation}
\label{eq:FJ_equilibrium}
\mathbf{z}^* = \left(I-(I-A)W\right)^{-1}A\mathbf{s} = V\mathbf{s},
\end{equation}
where $V = \left(I-(I-A)W\right)^{-1}A$ is the \emph{influence matrix}.
\end{theorem}


\subsection{The Poster-Observer Dynamics}

The classic FJ model operates under the assumption that all nodes synchronously ``post'' their opinion and ``observe'' the opinions of their neighbors at each time step. 
In the \emph{Poster-Observer FJ} (PO-FJ) model, which we introduce here, we decouple the processes of ``posting'' and ``observing'' for each node. 


Formally, the state of the network at each iteration $t$ is encoded by two binary \emph{activity vectors}: 
the \emph{observing} activity vector $\mathbf{x}^{(t)} \in \{0,1\}^n$ and the \emph{posting} activity vector $\mathbf{y}^{(t)} \in \{0,1\}^n$.
We set $x_i^{(t)}\!=1$ if agent $i$ observes their feed at iteration~$t$, and $x_i^{(t)}\!=0$ otherwise.
%
%
Similarly, we set $y_i^{(t)}\!=1$ if agent $i$ posts their opinion at iteration~$t$, and $y_i^{(t)}\!=0$ otherwise.
%

The activity vectors directly affect the graph, resulting in a time-dependent, row-stochastic \emph{Dynamic Interaction Matrix} $\hat{W}^{(t)}$:
\begin{equation}
\label{eq:Dynamic Interaction Matrix}
\hat{W}^{(t)} = X^{(t)}[WY^{(t)} + \text{diag}(W(\mathbf{1}-\mathbf{y}^{(t)}))] + (I - X^{(t)}),
\end{equation}
where $X^{(t)} = \text{diag}(\mathbf{x}^{(t)})$ and $Y^{(t)} = \text{diag}(\mathbf{y}^{(t)})$. 
Equivalently, we can think of the model as representing interactions over a dynamic graph where edges are added or deleted at every iteration based on user activity.
For the non-diagonal elements of $\hat{W}^{(t)}$ we have
$$ \hat{W}_{ij}^{(t)} = x_i^{(t)}y_j^{(t)}{W}_{ij}.$$

Note that the base interaction matrix edge weight $W_{ij}$ connects nodes $i$ and $j$, when both node $i$ observes \emph{and} node $j$ posts in iteration $t$. In any other case this edge does not exist in $\hat{W}^{(t)}$.

We enforce row-stochasticity for any resulting $\hat{W}^{(t)}$, via the diagonal elements which are equal to: $$\hat{W}_{ii}^{(t)} = x^{(t)}_i \left( W_{ii} + \sum_{k\neq i}\big(1-y^{(t)}_k\big)W_{ik} \right) + \big(1 - x^{(t)}_i\big).$$

Thus, when a user goes offline and stops observing ($x_i^{(t)}=0$), the entire $i$-th row of $\hat{W}^{(t)}$ becomes zero, except for the diagonal element $\hat{W}_{ii}^{(t)}=1$. 
This effectively removes all \emph{outward} edges from node $i$ for this iteration.
Conversely, when a user does not post ($y_i^{(t)}=0$), the $i$-th column of $\hat{W}^{(t)}$ becomes zero, except for the diagonal element $\hat{W}_{ii}^{(t)}$.
This effectively removes all \emph{inward} edges to node $i$ for this iteration.
In this case, the weights of the removed edges get moved to the diagonal elements of the matrix.
We believe that this transfer is realistic in the following sense: a user's total attention is fixed, and when they are not observing, their focus defaults back to their own previously expressed opinion ($\hat{W}_{ii}^{(t)}=1$). 
Similarly, when a user does not see a post from a neighbor, the weight that would have been spent reading that neighbor's post is transferred back to their own previously-expressed opinion.

With the dynamic interaction matrix $\hat{W}^{(t)}$ established, we can now define the PO-FJ opinion update process for our model.
\begin{definition}[The Poster-Observer Dynamics]
Given the dynamic interaction matrices $\hat{W}^{(t)}$, a diagonal stubbornness matrix $A$, and an inner opinion vector $\mathbf{s}$, the expressed opinions $\hat{\mathbf{z}}^{(t)}$ in the PO-FJ model evolve according to
\begin{equation}
\label{eq:PO_FJ_update}
\mathbf{\hat{z}}^{(t+1)} = A\mathbf{s} + (I-A)\hat{W}^{(t)}\mathbf{\hat{z}}^{(t)}.
\end{equation}
\end{definition}

Notice that  the updates on the PO-FJ expressed opinions~$\hat{\mathbf{z}}$ occur based on the dynamic interaction matrix $\hat{W}$ instead of the base interaction matrix $W$ of the classic FJ model (Eq.~\eqref{eq:FJ_update_vec}).
While the new update Equation~\eqref{eq:PO_FJ_update} can accommodate any arbitrary, time-varying sequence of communications, we choose a specific setup:


\begin{assumption}\label{ass:activity}
For any $t\in\{0,1,\ldots \}$
the activity vectors $\{\mathbf{x}^{(t)}\}_{t\ge0}$ and $\{\mathbf{y}^{(t)}\}_{t\ge0}$ are independent. 
Also they are independent and identically distributed over time $t$, with expectations
\begin{equation}
\label{eq:posting observing vectors}
\mathbb{E}\big[\mathbf{x}^{(t)}\big]=\mathbf{q}\in[0,1]^n,
  ~~~
\mathbb{E}\big[\mathbf{y}^{(t)}\big]=\mathbf{p}\in[0,1]^n ~~~ \mbox{ for all } t,
\end{equation}
and \emph{activity covariance matrices}
\begin{equation}
\label{eq:Sigma covariance matrices}
\operatorname{Cov}\big(\mathbf{x}^{(t)}\big)=\Sigma^x\in\mathbb{R}^{n\times n}, 
  ~~~
\operatorname{Cov}\big(\mathbf{y}^{(t)}\big)=\Sigma^y\in\mathbb{R}^{n\times n} ~~~ \mbox{ for all } t.
\end{equation}
\end{assumption}

Both $\Sigma^x$ and $\Sigma^y$ are symmetric, positive semidefinite matrices with diagonal values $\Sigma^x_{ii}=q_i(1-q_i)$ and $\Sigma^y_{ii}=p_i(1-p_i)$. We call $\mathbf{p}$ the \emph{posting probability vector} and $\mathbf{q}$ the \emph{observing probability vector}. 

Under this framework, we can model different kinds of user behaviors. 
An agent $i$ who posts frequently will have a high value~$p_i$, while an agent who mostly observes what others say without commenting will have a low value~$p_i$ but a high value~$q_i$.
Notice that for $\mathbf{p} = \mathbf{q} = \mathbf{1}$, we get the classic FJ model.

Correlations between agents posting and observing together in groups can be modeled via the \emph{activity covariance matrices}~$\Sigma^x$ and~$\Sigma^y$.
While more complex regimes could be used, such as activity vectors that correlate dynamically with the neighboring opinions or the network state, we will use the framework based on Assumption \ref{ass:activity} for this first application of the PO-FJ Model.

\subsection{On the Convergence of the PO-FJ Model}

The properties of the PO-FJ dynamics can be best understood through their analogy to the classic FJ model.
As established in Theorem~\ref{thm:fj_convergence}, classic FJ opinions deterministically converge to a static equilibrium $\mathbf{z}^*$.
In the PO-FJ model, the dynamic interaction matrix $\hat{W}^{(t)}$ prevents this static convergence: as network edges dynamically form and dissolve as users post and observe, opinions are prevented from ever locking into a single point.
Instead, one of our main results, established in Theorem \ref{thm:no_static_consensus}, shows that under PO-FJ dynamics the network reaches a \emph{probabilistic} equilibrium. 
We prove this in Appendix~\ref{app:no_static_consensus} using the iterated random functions work of \citet{Diaconis}.

\begin{theorem}[PO-FJ Convergence]\label{thm:no_static_consensus}
If matrix $W$ is \emph{irreducible}, the inner opinions are not trivially identical ($s_i \neq s_j$ for some $i,j$), and all agents have stubbornness $a_i \in (0,1)$ and stochastic activity $p_i, q_i \in (0,1)$, the sequence of opinions $\{\mathbf{\hat{z}}^{(t)}\}_{t\ge0}$ of the PO-FJ model never reaches a static equilibrium. 
Instead, the opinions converge to a unique, non-degenerate stationary probability distribution.
\end{theorem}

Theorem~\ref{thm:no_static_consensus} shows that under PO-FJ dynamics opinions $\hat{\mathbf{z}}^{(t)}$ follow stochastic paths that indefinitely oscillate, as visualized in Figure~\ref{fig:trajectory_comparison}. 
Under classic FJ (top), any initial opinion vector $\mathbf{z}^{(0)} \in [0,1]^n$ moves deterministically toward the static equilibrium $\mathbf{z}^*$ (red point).
Under PO-FJ dynamics (bottom), opinions $\hat{\mathbf{z}}^{(t)}$ follow stochastic trajectories that converge to a unique stationary probability distribution, depicted via its \emph{stationary expected opinion} value $\overline{\boldsymbol{\mu}}$ (red star) and \emph{stationary covariance} $\overline{M}$ (red shaded volume).
After reaching the distribution the opinions $\hat{\mathbf{z}}^{(t)}$ perpetually oscillate around it.
We next derive closed-form expressions for this trajectory.

\begin{figure}[t]
    \centering
    \includegraphics[width=0.9\columnwidth]{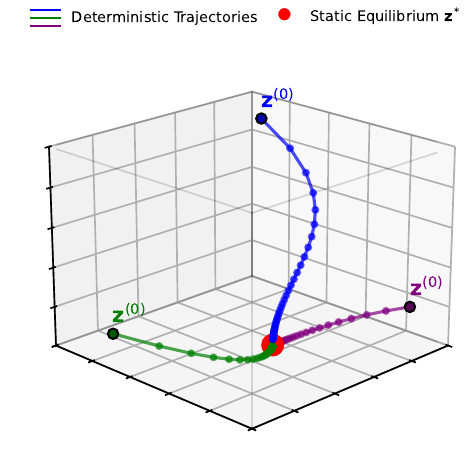}
    
    \vspace{2mm} 
    
    \includegraphics[width=0.9\columnwidth]{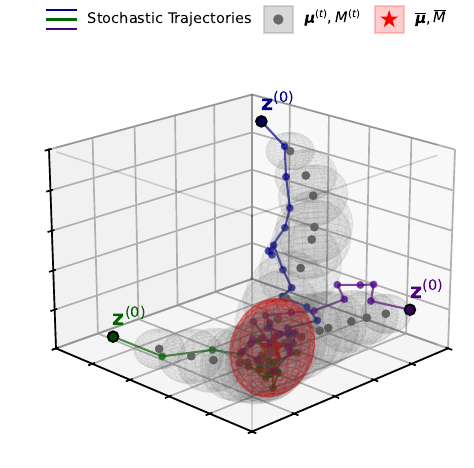}
    
    \caption{Opinion trajectories in the n-dimension hypercube $[0,1]^n$. \textbf{Top:} Classic FJ model with deterministic convergence. \textbf{Bottom:} PO-FJ model with stochastic opinion trajectories.}
    \label{fig:trajectory_comparison}
\end{figure}

\spara{Expected Opinions ($\boldsymbol{\mu}^{(t)}$):}
To analyze the expected trajectory of opinions under PO-FJ, we define the \emph{Expected Interaction Matrix} $\overline{W}$:

\begin{definition}[Expected Interaction Matrix]\label{def:expected_W}
The expectation of the dynamic interaction matrix $\overline{W} := \mathbb{E}\big[\hat{W}^{(t)}\big]$ is given by:
\[
  \overline{W} = \operatorname{diag}(\mathbf{q}) \Big[\, W \operatorname{diag}(\mathbf{p}) + \operatorname{diag}\!\big(W(\mathbf{1}-\mathbf{p})\big) \,\Big] + \big(I - \operatorname{diag}(\mathbf{q})\big).
\]
\end{definition}
\noindent By construction, $\overline{W}$ remains row-stochastic. Furthermore, provided that the base interaction matrix $W$ is strongly connected and $q_i, p_i > 0$ for all nodes, $\overline{W}$ inherits the strong connectivity and aperiodicity of $W$. By taking the expectation of the opinion update rule (Eq.~\eqref{eq:PO_FJ_update}), we establish the following result in Appendix~\ref{app:Lemma1}:

\begin{lemma}[Expected Opinions]\label{lm:expected_opinions}
The expected opinions defined as $\boldsymbol{\mu}^{(t)} := \mathbb{E}\big[\mathbf{\hat{z}}^{(t)}\big]$ at each iteration $t$ evolve according to:
\begin{equation}\label{eq:meanrec}
  \boldsymbol{\mu}^{(t+1)} = A \mathbf{s} + (I - A)\overline{W}\,\boldsymbol{\mu}^{(t)}.
\end{equation}
This trajectory converges to the \emph{stationary expected opinion} vector:
\begin{equation}\label{eq:meanss}
  \overline{\boldsymbol{\mu}} := \lim_{t \to \infty} \boldsymbol{\mu}^{(t)} = \big(I - (I-A)\overline{W}\big)^{-1}A\,\mathbf{s} = \overline{V}\,\mathbf{s},
\end{equation}
where $\overline{V} := \big(I - (I-A)\overline{W}\big)^{-1}A$ is the expected influence matrix.
\end{lemma}

Lemma~\ref{lm:expected_opinions} shows that while the actual opinions $\mathbf{\hat{z}}^{(t)}$ undergo perpetual stochastic oscillations, their expectation $\boldsymbol{\mu}^{(t)}$ follows classic FJ dynamics under the expected interaction matrix $\overline{W}$. 
Since $\overline{W}$ preserves the sparsity of $W$, calculating $\overline{\boldsymbol{\mu}}$ can be efficiently done by sequentially applying Equation~\eqref{eq:meanrec}, similarly to the classic FJ.

\spara{Opinion Covariance ($M^{(t)}$):}
To characterize the amplitude and correlation of the opinion oscillations around $\boldsymbol{\mu}^{(t)}$, we define the \emph{opinion covariance matrix} $M^{(t)}$ at iteration $t$ as:
\begin{equation}
    M^{(t)} := \operatorname{Cov}(\mathbf{\hat{z}}^{(t)}) = \mathbb{E} \Big[ \big(\mathbf{\hat{z}}^{(t)} - \boldsymbol{\mu}^{(t)}\big) \big(\mathbf{\hat{z}}^{(t)} - \boldsymbol{\mu}^{(t)}\big)^T \Big].
\end{equation}
The diagonal elements $M^{(t)}_{ii} = \operatorname{Var}(\hat{z}_{i}^{(t)})$ quantify the variance of each user's opinion, while the off-diagonal elements $M^{(t)}_{ij} = \operatorname{cov}(\hat{z}_{i}^{(t)},\hat{z}_{j}^{(t)})$ capture the co-movement between users. 
For instance, if the opinions of two users $i$ and $j$ tend to oscillate together in the same direction, they will exhibit a high $M^{(t)}_{ij}$ value.

To formalize the evolution of $M^{(t)}$, we define the following operators: $B_\mathbf{v} := W\operatorname{diag}(\mathbf{v}) - \operatorname{diag}(\mathbf{v})W$ for any vector $\mathbf{v}$, and $\mathcal{G}(C, D) := W(C \circ D)W^T - W\big(C \circ (DW^T)\big) - \big((WD) \circ C\big)W^T + C \circ (WDW^T)$ for any matrices $C$ and $D$. We prove in Appendix~\ref{app:Lemma2} that the opinion covariance matrix $M^{(t)}$ evolves as follows:

\begin{lemma}[Opinion Covariance Matrix]\label{lm:covariance}
The opinion covariance matrix $M^{(t)}$ at each iteration $t$ evolves according to:
\begin{equation}\label{eq:covrec}
  M^{(t+1)} = F\big(M^{(t)}\big) + \Omega^{(t)},
\end{equation}
where the operator $F(M^{(t)})$ is defined as:
\begin{equation*}
    F(M^{(t)}) = (I-A) \Big[ \overline{W} M^{(t)} \overline{W}^T + \mathcal{H}(M^{(t)}) \Big] (I-A),
\end{equation*}
and $\Omega^{(t)}$ is defined as:
\begin{equation*}
    \Omega^{(t)} = (I-A) \Big[ (\Sigma^x + \mathbf{q}\mathbf{q}^T) \circ \big( B_{\boldsymbol{\mu}^{(t)}} \Sigma^y B_{\boldsymbol{\mu}^{(t)}}^T \big) + \Sigma^x \circ \big( B_{\boldsymbol{\mu}^{(t)}} \mathbf{p}\mathbf{p}^T B_{\boldsymbol{\mu}^{(t)}}^T \big) \Big] (I-A).
\end{equation*}
The function $\mathcal{H}(M^{(t)})$ is evaluated as:
\begin{equation*}
    \mathcal{H}(M^{(t)}) = (\Sigma^x + \mathbf{q}\mathbf{q}^T) \circ \mathcal{G}\big(M^{(t)}, \Sigma^y\big) + \Sigma^x \circ \mathcal{G}\big(M^{(t)}, \mathbf{p}\mathbf{p}^T\big).
\end{equation*}
\end{lemma}

Lemma~\ref{lm:covariance} shows that covariance updates through two distinct forces: $F(M^{(t)})$ propagates and decays the existing variance, while $\Omega^{(t)}$ continuously injects new variance.
Both of these forces depend not only on the activity probabilities ($\mathbf{p}, \mathbf{q}$) but also on the activity covariance ($\Sigma^x, \Sigma^y$), showing that synchronized user activity can directly affect the opinion oscillations.
By the convergence result of Theorem~\ref{thm:no_static_consensus}, applying iteratively Equation~\eqref{eq:covrec} converges to the unique \emph{stationary opinion covariance matrix} $\overline{M} := \lim_{t\to \infty}M^{(t)}$.


\section{Conflict Limitation via Exposure and Amplification Regulation (CLEAR)}
\label{sec:metrics}
\label{sec:mitigation}

A central goal in the opinion-dynamics literature is to quantify and control network-level phenomena such as polarization and disagreement \cite{Howbad, matakos2017measuring, Musco,kuhne2025optimizing}. 
In this section, we first establish how these classic metrics are generalized under the PO-FJ dynamics. 
Then, we introduce the CLEAR problem of mitigating network conflict by controlling the new levers of the PO-FJ model.

\subsection{Polarization and Disagreement under PO-FJ}

In classic FJ opinion dynamics, \emph{Public Opinion} ($\mathcal{Z}_{\text{FJ}}$), \emph{Polarization} ($\mathcal{P}_{\text{FJ}}$), and \emph{Disagreement} ($\mathcal{D}_{\text{FJ}}$) are well-established metrics \cite{Howbad, kuhne2025optimizing,Musco} for evaluating the state of a network:
\begin{equation}\label{eq:po_classic}
\mathcal{Z}_{\text{FJ}} = \frac{1}{n} \sum_{i=1}^n z_i^*,
\end{equation}
\begin{equation}\label{eq:pol_classic}
\mathcal{P}_{\text{FJ}} = \frac{1}{n} \sum_{i=1}^n \left(z_i^* - \mathcal{Z}_{\text{FJ}}\right)^2, \mbox{ and}
\end{equation}
\begin{equation}\label{eq:dis_classic}
\mathcal{D}_{\text{FJ}} = \frac{1}{2n} \sum_{i=1}^n \sum_{j=1}^n W_{ij} \left(z_i^* - z_j^*\right)^2. 
\end{equation}
In the classic setting, these metrics are evaluated at the static equilibrium $\mathbf{z}^*$ from Eq.~\eqref{eq:FJ_equilibrium}. 
However, as established in Theorem~\ref{thm:no_static_consensus}, the PO-FJ model prevents convergence to a static point and causes opinions to perpetually oscillate. 
Consequently, at any time\-step~$t$, these quantities become random variables. We define them using the PO-FJ expressed opinions $\mathbf{\hat{z}}^{(t)}$ and dynamic interactions~$\hat{W}^{(t)}_{ij}$:
\begin{equation}\label{eq:po_POFJ}
\mathcal{Z}^{(t)} := \frac{1}{n} \sum_{i=1}^n \hat{z}_i^{(t)},
\end{equation}
\begin{equation}\label{eq:pol_POFJ}
\mathcal{P}^{(t)} := \frac{1}{n} \sum_{i=1}^n \left(\hat{z}_i^{(t)} - \mathcal{Z}^{(t)}\right)^2, \mbox{ and}
\end{equation}
\begin{equation}\label{eq:dis_POFJ}
\mathcal{D}^{(t)} := \frac{1}{2n} \sum_{i=1}^n \sum_{j=1}^n \hat{W}^{(t)}_{ij} \left(\hat{z}_i^{(t)} - \hat{z}_j^{(t)}\right)^2. 
\end{equation}

$\mathcal{Z}^{(t)}$, $\mathcal{P}^{(t)}$, and $\mathcal{D}^{(t)}$ fluctuate at every iteration $t$.
However, because Theorem~\ref{thm:no_static_consensus} guarantees that the opinions $\{\mathbf{\hat{z}}^{(t)}\}_{t \ge 0}$ converge to a stationary probability distribution, we can show that the expected values of these metrics also converge to stationary limits. 
Thus, to evaluate the state of the network we take the limit of the expectation of these values, and we define the \emph{Stationary Expected Public Opinion} ($\overline{\mathcal{Z}} := \lim_{t \to \infty}\mathbb{E} [ \mathcal{Z}^{(t)} ]$), \emph{Stationary Expected Polarization} ($\overline{\mathcal{P}} := \lim_{t \to \infty}\mathbb{E}[\mathcal{P}^{(t)}]$), and \emph{Stationary Expected Disagreement} ($\overline{\mathcal{D}} := \lim_{t \to \infty}\mathbb{E}[\mathcal{D}^{(t)}]$). 
We prove in Appendix~\ref{app:PolDis} that:
\begin{equation}\label{eq:po_POFJ_stas}
\overline{\mathcal{Z}} = \frac{1}{n} \sum_{i=1}^n \overline{\mu}_i,
\end{equation}
\begin{equation}\label{eq:pol_POFJ_stas}
\overline{\mathcal{P}} = \frac{1}{n} \sum_{i=1}^n \left(\overline{\mu}_i - \overline{\mathcal{Z}} \right)^2 + \frac{1}{n} \sum_{i=1}^n \overline{M}_{ii} - \frac{1}{n^2} \sum_{i=1}^n \sum_{j=1}^n \overline{M}_{ij},
\end{equation}
\begin{equation}\label{eq:dis_POFJ_stas}
\begin{aligned}
\overline{\mathcal{D}}  = &\frac{1}{2n} \sum_{i=1}^n \sum_{j=1}^n \overline{W}_{ij} (\overline{\mu}_i - \overline{\mu}_j)^2 \\ 
& + \frac{1}{2n} \sum_{i=1}^n (1 + \overline{d}_i^{\text{in}}) \overline{M}_{ii} - \frac{1}{n} \sum_{i=1}^n \sum_{j=1}^n \overline{W}_{ij} \overline{M}_{ij},
\end{aligned}
\end{equation}
where $\overline{d}_i^{\text{in}} = \sum_k \overline{W}_{ki}$ is the expected in-degree of node $i$.
Equations~\eqref{eq:pol_POFJ_stas} and~\eqref{eq:dis_POFJ_stas} reveal a fundamental difference between the PO-FJ metrics and those of the classic FJ model. 
Both $\overline{\mathcal{P}}$ and $\overline{\mathcal{D}}$ decompose into two parts: ($i$) a ``classic'' term evaluated at the stationary expected opinions $\overline{\boldsymbol{\mu}}$ and expected interaction matrix $\overline{W}$, plus ($ii$) two new terms governed by the stationary opinion covariance matrix $\overline{M}$. 
These covariance terms show that individual opinion variances ($\overline{M}_{ii} > 0$) strictly increase both metrics: the more intensely nodes oscillate around their expected opinions under PO-FJ, the higher the resulting polarization and disagreement. 
At the same time, positive covariances between nodes ($\overline{M}_{ij} > 0$) reduce both metrics when users' opinions fluctuate together in the same direction. 
Consequently, even if two connected nodes share the exact same expected opinion ($\overline{\mu}_i = \overline{\mu}_j$), classic FJ would report zero disagreement, whereas PO-FJ captures the disagreement that arises when their opinions oscillate in opposite directions.

\subsection{The CLEAR Problem}

Existing literature on controlling the metrics defined in Eqs.~\eqref{eq:po_classic}, \eqref{eq:pol_classic}, \eqref{eq:dis_classic} focuses on modifying classic FJ parameters, such as changing graph edges or shifting inner opinions \cite{gionis, matakos2020tell, Tsourakakis,  Xu2022}. 
However, directly altering a social graph or inner opinions is unrealistic in practice, as network controllers cannot dictate whom users follow or what they believe in.
Instead, actual control in modern social media is mediated by recommendation algorithms: the platform can choose to \emph{amplify} or \emph{throttle} a user. 
This can be modeled under the PO-FJ dynamics, since adjustments in $\mathbf{p}$ and $\mathbf{q}$ have a natural interpretation in modern web platforms: vector $\mathbf{p}$ models algorithmic amplification of a user, while vector $\mathbf{q}$ models the exposure of a node to its surroundings.

To adjust both polarization and disagreement, we define the overall \emph{Conflict} of the network under the PO-FJ model as:
\begin{equation}\label{eq:conflict}
\overline{\mathcal{C}} = \overline{\mathcal{P}} + \overline{\mathcal{D}}.
\end{equation}
We note that a weighted version of conflict can be introduced, as discussed by \citet{Musco}, by weighting differently the two terms 
depending on application requirements. 
We focus on the standard unweighted sum in Equation~\eqref{eq:conflict} for this first study.

By treating the vectors $\mathbf{p}, \mathbf{q} \in [0,1]^n$ as tunable parameters with an intervention $L_1$ budget $b$, we formulate the Conflict Limitation via Exposure and Amplification Regulation (CLEAR) problem:

\begin{problem}[CLEAR]\label{prob:clear}
Given a PO-FJ network and initial vectors $\mathbf{p}^{(0)}, \mathbf{q}^{(0)} \in [0,1]^n$, find the adjusted vectors $\mathbf{p}$ and $\mathbf{q}$ that minimize conflict, subject to an intervention budget $b$:
\begin{equation}
\begin{aligned}
    \min_{\mathbf{p}, \mathbf{q}} \quad & \overline{\mathcal{C}}(\mathbf{p}, \mathbf{q}) \\
    \text{such that} \quad & \|\mathbf{p} - \mathbf{p}^{(0)}\|_1 + \|\mathbf{q} - \mathbf{q}^{(0)}\|_1 \le b, \\
    & \mathbf{0} \le \mathbf{p} \le \mathbf{1}, \quad \mathbf{0} \le \mathbf{q} \le \mathbf{1}.
\end{aligned}
\end{equation}
\end{problem}

The objective $\overline{\mathcal{C}}(\mathbf{p}, \mathbf{q})$ is a highly complex function that depends on $\mathbf{p}$ and $\mathbf{q}$ through $\overline{W}$, $\overline{\boldsymbol{\mu}}$, and $\overline{M}$. 
While finding the global optimum solution for the CLEAR problem is difficult, we can compute reliable solutions via Projected Gradient Descent (PGD), as detailed in Algorithm~\ref{alg:clear_pgd} (\text{CLEAR-PGD}).
A key challenge in applying gradient-based optimization in our setup is computing the derivatives of $\overline{\mathcal{C}}$ with respect to $\mathbf{p}$ and $\mathbf{q}$.
Differentiating at the closed-form equilibrium equations is doable for $\overline{\boldsymbol{\mu}}$ but \emph{not} for $\overline{M}$.
To overcome this, we treat the iterative update equations ~\eqref{eq:meanrec} and~\eqref{eq:covrec} for $\boldsymbol{\mu}^{(t)}$ and $M^{(t)}$ as a recurrent sequence, similar to the forward pass in classic neural network training.
We first unroll these dynamics for $T_{\mu}$ and $T_{M}$ mixing steps until they converge to $\overline{\boldsymbol{\mu}}$, $\overline{M}$, and thus $\overline{\mathcal{C}}$ (Step 1).
Next we execute Backpropagation through the unrolled steps to compute $\nabla_{\mathbf{p}} \overline{\mathcal{C}}$ and $\nabla_{\mathbf{q}} \overline{\mathcal{C}}$ (Step 2).
Backpropagation through classic FJ dynamics was recently used in \cite{kuhne2025optimizing}, and we now apply it for the expected opinions and the opinion covariance matrix.
We then update $(\mathbf{p}, \mathbf{q})$ (Step 3) and project the updated vectors back onto the $[0,1]$ box bounds and the $L_1$ budget (Step 4), repeating until the change in conflict falls below a tolerance $\tau$.

\spara{Computational Complexity:} 
The Forward Pass (Step 1) requires evaluating $\boldsymbol{\mu}^{(t)}$ and $M^{(t)}$ across $T_{\mu}$ and $T_{M}$ mixing steps, respectively. 
Because matrix $\overline{M}$ is dense, we apply a $K$-hop approximation that retains entries between nodes at most $K$ hops apart (See Appendix~\ref{app:khop_approx}); setting $K=1$ results in $\mathcal{O}((T_{\mu}+T_{M}) |E|)$ complexity for Step~1.
The Backward Pass (Step 2) computes gradients $\nabla_{(\mathbf{p},\mathbf{q})} \overline{\mathcal{C}}$, and has the same complexity as Step 1. 
The Update in Step 3 is an element-wise vector operation taking $\mathcal{O}(n)$ time, while the projection in Step 4 is solved via a sorting-based $L_1$-projection in~$\mathcal{O}(n \log n)$ time. 
Thus, each PGD iteration is dominated by the forward and backward passes, yielding an overall complexity for Algorithm~\ref{alg:clear_pgd} of $\mathcal{O}(T_{\text{opt}} \left( (T_{\mu} + T_{M})  |E| + n\log n)\right)$. 
As demonstrated empirically in Section~\ref{sec:experiments}, the mixing ($T_{\mu}, T_{M}$) and optimization steps ($T_{\text{opt}}$) remain small in practice, ensuring a scalable solution for Problem \ref{prob:clear}.

\begin{algorithm}[t]
\caption{CLEAR-PGD}
\label{alg:clear_pgd}
\begin{algorithmic}[1]
\State \textbf{Input:} Base interaction matrix $W$, initial vectors $\mathbf{p}^{(0)}, \mathbf{q}^{(0)}$, inner opinions $\mathbf{s}$, stubbornness $\mathbf{a}$, budget $b$, tolerance $\tau$
\State \textbf{Output:} Optimized $\mathbf{p}$ and $\mathbf{q}$ vectors
\State Initialize $\mathbf{p} \leftarrow \mathbf{p}^{(0)}$, $\mathbf{q} \leftarrow \mathbf{q}^{(0)}$, $\overline{\mathcal{C}}^{(0)} \leftarrow \infty$
\For{$t = 1$ to $T_{\text{opt}}$}
    \State \textbf{Step 1: Forward Pass}
    \State Compute $\overline{\boldsymbol{\mu}}$, $\overline{M}$ via Eqs.~\eqref{eq:meanrec}, \eqref{eq:covrec} and conflict $\overline{\mathcal{C}}^{(t)}$.
    \If{$|\overline{\mathcal{C}}^{(t-1)} - \overline{\mathcal{C}}^{(t)}| < \tau$} \textbf{break} \EndIf
    \State \textbf{Step 2: Backward Pass}
    \State Compute gradients $\nabla_{(\mathbf{p},\mathbf{q})} \overline{\mathcal{C}}^{(t)}$ via Backpropagation.
    \State \textbf{Step 3: Update}
    \State $(\tilde{\mathbf{p}}, \tilde{\mathbf{q}}) \leftarrow (\mathbf{p}, \mathbf{q}) - \eta \nabla_{(\mathbf{p}, \mathbf{q})} \overline{\mathcal{C}}^{(t)}$
    \State \textbf{Step 4: Projection}
    \State Enforce $L_1$ budget and $[0,1]$ box bounds:
    \State $(\mathbf{p}, \mathbf{q}) \leftarrow \text{Project}(\tilde{\mathbf{p}}, \tilde{\mathbf{q}}, \mathbf{p}^{(0)}, \mathbf{q}^{(0)}, b)$
\EndFor
\State \Return $\mathbf{p}, \mathbf{q}$
\end{algorithmic}
\end{algorithm}

We also consider a combinatorial variant of our problem, which we call \emph{Discrete-CLEAR}, where the platform is given an integer budget $k < n$ and selects at most $k$ nodes to either completely mute or fully activate ($p_i, q_i \in \{0,1\}$) in order to minimize conflict. 
In Appendix~\ref{app:nphard}, we formally define this problem and prove via a reduction from Vertex Cover that it is NP-hard. Designing combinatorial algorithms for the Discrete-CLEAR problem is left as future work.

\section{Experiments}\label{sec:experiments}

In our experiments, we evaluate the PO-FJ model on real-world and synthetic networks to address the Research Questions (RQs): 

\begin{enumerate}[label=RQ\arabic*:]
    \item How do the posting and observing behaviors of users affect public opinion, polarization and disagreement?
    \item How effective is the CLEAR-PGD solution in the CLEAR problem compared to baselines?
    \item Which nodes does CLEAR-PGD alter and how?
    \item How does CLEAR-PGD scale?
\end{enumerate}

\spara{Datasets and Hardware:} 
Table~\ref{tab:datasets_main} summarizes the Web datasets from SNAP \cite{snapnets} we use. 
To ensure stationary FJ dynamics, we restrict all networks to their Largest Strongly Connected Component (LSCC). 
We construct the base interaction matrix $W$ using row-stochastic uniform weighting $W_{ij} = 1/d_i^{\text{out}}$, where $d_i^{\text{out}}$ is the out-degree of node $i$. 
We execute all experiments on a single NVIDIA A100 Core GPU with 40 GB of VRAM.
Source code available  online.\footnote{\url{https://anonymous.4open.science/r/PO-FJ-Opinion-Model-58E1/}}

\begin{table}[t]
    \centering
    \caption{Summary of network datasets.}
    \label{tab:datasets_main}
    \begin{tabular}{l r r l}
        \toprule
        \textbf{Dataset} & \textbf{Nodes ($n$)} & \textbf{Edges ($|E|$)} & \textbf{Type} \\
        \midrule
        Wiki & 1,300 & 39,456 & Directed \\
        Gnutella & 2,068 & 9,313 & Directed \\
        Epinions & 32,223 & 443,506 & Directed \\
        Slashdot & 71,307 & 912,381 & Directed \\
        YouTube & 1,134,890 & 5,975,248 & Undirected \\
        Pokec & 1,304,537 & 29,183,655 & Directed \\
        \bottomrule
    \end{tabular}
\end{table}

\spara{Algorithm \ref{alg:clear_pgd} (CLEAR-PGD) Parameters:} 
For Step 3, we use learning rate $\eta = \eta_0 n$, where after grid search we set $\eta_0 = 0.5$ and $n$ is the number of nodes. 
We set mixing steps $T_\mu = 40$ and $T_M = 15$, which empirically guarantee convergence to $\overline{\boldsymbol{\mu}}$ and $\overline{M}$ across all graphs. 
We apply $K$-hop structural truncation on $M^{(t+1)}$ with $K=1$, preventing $\overline{M}$ from densifying while capturing the majority of its information (See Appendix~\ref{app:khop_approx}). We set tolerance $\tau=10^{-5}$ ($T_{\text{opt}}=25$).



\subsection{RQ1: Impact of posting and observing behaviors on opinion formation metrics}

To answer RQ1, we construct a Stochastic Block Model graph with $100$ nodes, partitioned equally into a Blue community ($s_i = 0$) and a Red community ($s_i = 1$). 
Edge probabilities are set to $0.15$ for intra-community and $0.02$ for inter-community connections. 
The stubbornness of all nodes is set to $a_i = 0.5$. 

Under classic FJ ($p = q = \mathbf{1}$), this network results in: ${\mathcal{Z}}_{\text{FJ}} = 0.5159$, ${\mathcal{P}}_{\text{FJ}} = 0.1612$, ${\mathcal{D}}_{\text{FJ}} = 0.0376$. 
To show how PO-FJ affects these results, we fix the Blue community to an average posting and observing activity ($p_{\text{blue}} = 0.5, q_{\text{blue}} = 0.5$). 
We then vary the behavior of the Red community across a grid of posting ($p_{\text{red}} \in [0,1]$) and observing ($q_{\text{red}} \in [0,1]$) probabilities. 
The activity covariance matrices are set diagonal: $\Sigma^x = \operatorname{diag}(\mathbf{q} \circ (\mathbf{1}-\mathbf{q}))$ and $\Sigma^y = \operatorname{diag}(\mathbf{p} \circ (\mathbf{1}-\mathbf{p}))$.

\begin{figure*}[t]
    \centering
    \includegraphics[width=\textwidth]{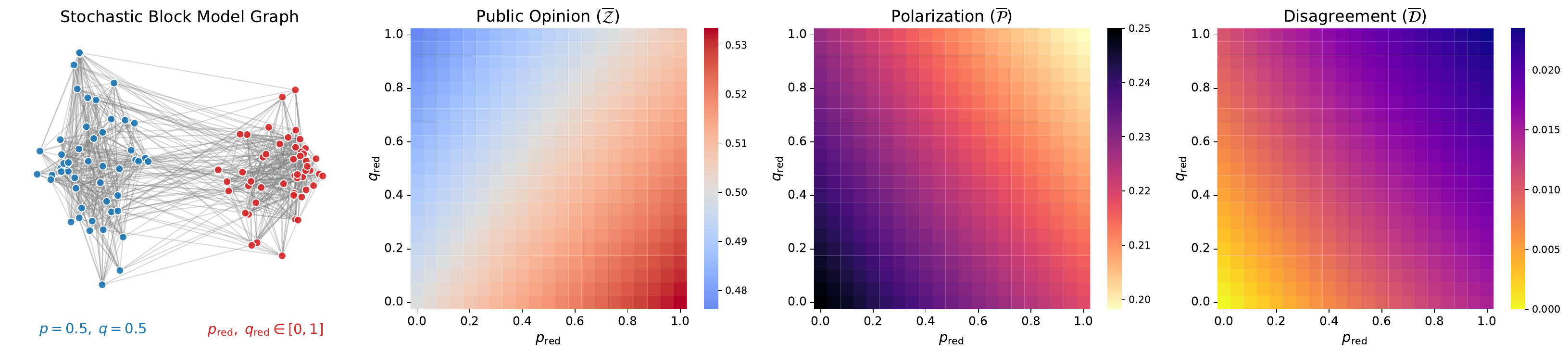}
    \caption{The effect of changing posting ($p_{\text{red}}$) and observing ($q_{\text{red}}$) behaviors on an SBM network.}
    \label{fig:heatmap_sbm}
\end{figure*}

As shown in Figure~\ref{fig:heatmap_sbm}, altering $p_{\text{red}}$ and $q_{\text{red}}$ shifts public opinion, polarization, and disagreement in ways that classic FJ fails to capture. 
When the Red community acts as ``lurkers'', observing while rarely posting (top-left of the heatmaps), it absorbs the Blue community's views without broadcasting its own. 
This pulls the public opinion toward the Blue community's opinion ($\overline{\mathcal{Z}} \approx 0.48$).
Conversely, when the Red community acts as ``broadcasters'', posting heavily while rarely observing (bottom-right of the heatmaps), they refuse to absorb the Blue community's signals, shifting public opinion toward the Red group ($\overline{\mathcal{Z}} \approx 0.53$).

These behavioral shifts also reveal a trade-off between polarization and disagreement in their respective heatmaps. 
When the Red community does not post or observe ($p_{\text{red}} =  q_{\text{red}} = 0$, bottom-left), Red nodes remain isolated in their inner beliefs, which maximizes network polarization ($\overline{\mathcal{P}} \approx 0.25$). 
However, because non-observant nodes do not engage with opposing views, disagreement is minimized in this case ($\overline{\mathcal{D}} \approx 0$). 
In contrast, as the Red community increasingly posts and observes ($p_{\text{red}} = q_{\text{red}} = 1$, top-right of the heatmaps), cross-community interaction draws the two groups closer together to reduce polarization to its minimum ($\overline{\mathcal{P}} \approx 0.2$). However, this exposes opposing opinions to each other causing disagreement to reach its maximum ($\overline{\mathcal{D}} \approx 0.02$).

The example of Figure~\ref{fig:heatmap_sbm} demonstrates that different posting and observing behaviors can dictate the final state of the network. 
This motivated us to test how intervention algorithms can incorporate the Poster-Observer dynamics when attempting to mitigate conflict.

\subsection{RQ2: How effective is our solution?}

We evaluate the effectiveness of solving CLEAR by applying Algorithm~\ref{alg:clear_pgd} to the networks of Table~\ref{tab:datasets_main}. 
To ensure our evaluation is robust to different initializations, we use two experimental setups:

\begin{itemize}[leftmargin=*]
    \item \textbf{Scenario A:} We partition the network into two distinct communities using Spectral Bisection \cite{pothen1990partitioning}. We assign inner opinions $s_i \sim \mathcal{U}(0, 0.3)$ to one half and $s_i \sim \mathcal{U}(0.7, 1.0)$ to the other.
    \item \textbf{Scenario B:} We assign inner opinions uniformly at random, $s_i \sim \mathcal{U}(0, 1)$, across the nodes of the network.
\end{itemize}

To model realistic platform activity \cite{nielsen2006participation,CarronArthur2014}, we model users of three categories. Ranking by their PageRank values ($\pi_i$), the top 1\% of nodes act as broadcasters ($p_i^{(0)} = 0.9, q_i^{(0)} = 0.1$), the next 9\% act as contributors ($p_i^{(0)} = 0.5, q_i^{(0)} = 0.5$), and the remaining 90\% act as observers ($p_i^{(0)} = 0.1, q_i^{(0)} = 0.9$). 
The activity covariance matrices are set diagonal: $\Sigma^x = \operatorname{diag}(\mathbf{q} \circ (\mathbf{1}-\mathbf{q}))$ and $\Sigma^y = \operatorname{diag}(\mathbf{p} \circ (\mathbf{1}-\mathbf{p}))$.
In both scenarios, we set all users' stubbornness at $a_i = 0.5$ and we compare our solution against the following baselines:

\begin{itemize}[leftmargin=*]
    \item \textbf{Random:} Spends the intervention budget $b$ on uniform random perturbations across the $\mathbf{p}$ and $\mathbf{q}$ vectors.
    \item \textbf{Ideological-Heuristic:} Alters users based on their ideological extremity and PageRank centrality $\pi_i$. Converts influential extremists ($\max_i \pi_i |s_i - 0.5|$) into observers ($p_i = 0, q_i = 1$) and influential moderates ($\max_i \pi_i (0.5 - |s_i - 0.5|)$) into broadcasters ($p_i = 1, q_i = 0$), sequentially until budget $b$ is exhausted.
    \item \textbf{Bridge-Heuristic:} Identifies connected pairs $(i,j) \in E$ with opposing views and maximizes their observing behaviors ($q_i = q_j = 1$). Prioritizes edges using a PageRank-weighted ideological score, $(\pi_i + \pi_j)|s_i - s_j|$, exhausting the budget $b$ iteratively.
    \item \textbf{CLEAR-PGD ($\mathbf{p}$-only):} A version of Algorithm~\ref{alg:clear_pgd} where the platform is only allowed to alter vector $\mathbf{p}$.
    \item \textbf{CLEAR-PGD ($\mathbf{q}$-only):} A version of Algorithm~\ref{alg:clear_pgd} where the platform is only allowed to alter vector $\mathbf{q}$.
    \item \textbf{CLEAR-PGD:} Algorithm~\ref{alg:clear_pgd}.
\end{itemize}

\begin{figure*}[t]
    \centering
    \vspace{-6pt}
    \includegraphics[width=0.98\textwidth]{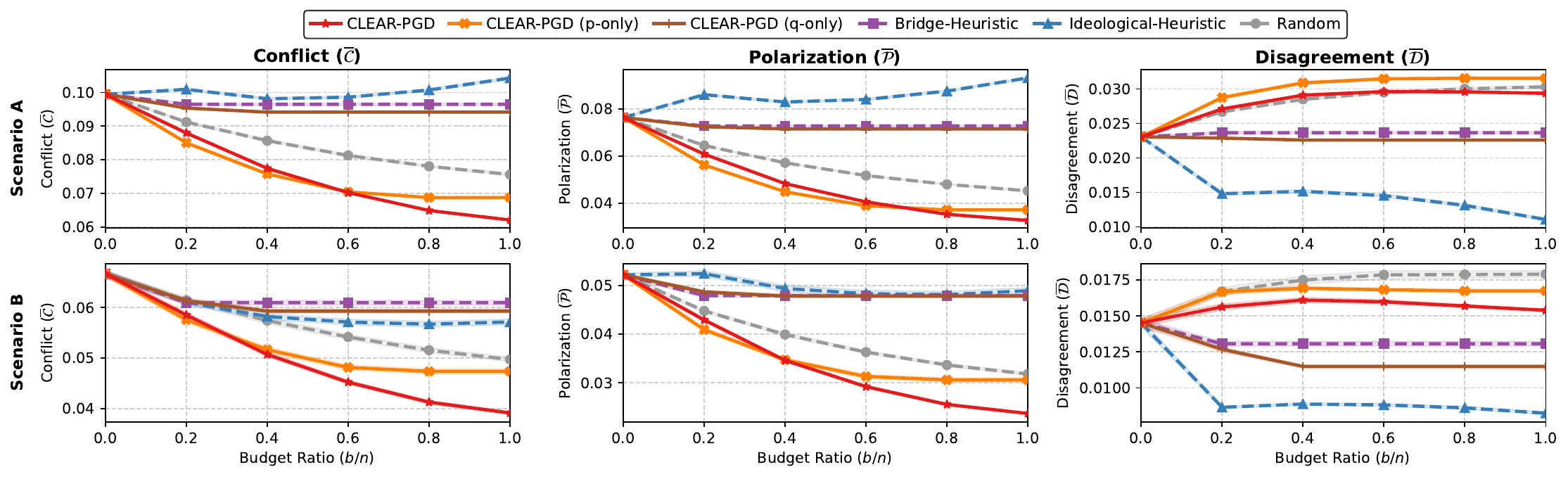}
    \vspace{-8pt}
    \caption{Performance of Algorithm \ref{alg:clear_pgd} and Baselines on Problem \ref{prob:clear} for the Epinions graph across different intervention budgets.}
    \label{fig:rq2_epinions}
    \vspace{-10pt}
\end{figure*}

Figure~\ref{fig:rq2_epinions} shows the resulting conflict ($\overline{\mathcal{C}}$), polarization ($\overline{\mathcal{P}}$), and disagreement ($\overline{\mathcal{D}}$) across different budget ratios $b/n$, reporting the mean and standard deviation over five random initializations.
In Figure~\ref{fig:rq2_epinions} we show results for the Epinions graph. 
The rest of the graphs of Table \ref{tab:datasets_main} show similar results in Table \ref{tab:appendix_rq2} of Appendix \ref{app:Experiments}.
Across all budgets and both scenarios, CLEAR-PGD consistently dominates all baseline methods in minimizing network conflict.
Comparing the two CLEAR-PGD variants ({$\mathbf{p}$-only vs. $\mathbf{q}$-only)} shows that adjustments made only on vector $\mathbf{q}$ do not perform nearly as well as adjustments made only on vector $\mathbf{p}$.  
This shows that controlling user \emph{amplification} acts as a much more powerful lever for conflict mitigation than controlling their exposure, though joint optimization achieves the best reduction.

The heuristic results show how simple intervention rules can decrease but also increase network conflict.
Figure \ref{fig:rq2_epinions} shows that the Bridge-Heuristic consistently reduces conflict across both scenarios, though it plateaus early and remains strictly suboptimal compared to CLEAR-PGD. 
Interestingly, the Ideological-Heuristic exhibits different results depending on the underlying opinion distribution.
In Scenario B, where moderate opinions are present, amplifying central moderates while muting extremists steadily reduces conflict.
However, in Scenario A (polarized communities), where no moderates exist, forcing the least extreme nodes ($s_i \approx 0.3, 0.7$) to become broadcasters does not have the desired effect. 
While this lowers disagreement, it locks the two communities into their opposing views and drives up polarization and the overall conflict.
Applying random perturbations also achieves modest conflict reduction.
By disrupting the 90-9-1 posting and observing hierarchy, where very few nodes speak, Random allows a broader set of nodes to interact. 
This reduces polarization; however, because these shifts are uncoordinated, disagreement is also increased. 

Overall, these findings reveal that simple heuristic adjustments, such as bridging opposing nodes, or changing extremists to observers and moderates to broadcasters, do reduce conflict to some extent but not as much as Algorithm \ref{alg:clear_pgd}.
Our CLEAR-PGD solution achieves the best performance, and in the RQ3 subsection, we analyze the adjustments it performs to reach this state.

\spara{Trade-off Between $\overline{\mathcal{P}}$ and $\overline{\mathcal{D}}$.}
As shown both in Figures~\ref{fig:heatmap_sbm} and~\ref{fig:rq2_epinions}, polarization $\overline{\mathcal{P}}$ has a higher numeric value than disagreement $\overline{\mathcal{D}}$. 
Because Problem~\ref{prob:clear} minimizes conflict  ($\overline{\mathcal{C}} = \overline{\mathcal{P}} + \overline{\mathcal{D}}$), Algorithm \ref{alg:clear_pgd} naturally prioritizes large reductions in polarization, accepting minor increases in disagreement as it exposes polarized
users to opposing opinions. 
As noted in Section~\ref{sec:metrics}, platforms that wish to penalize local disagreement more heavily can introduce a relative weight in Equation~\eqref{eq:conflict} that would change the focus of conflict to be more concentrated in the disagreement term $\overline{\mathcal{D}}$.
This does not alter the application or complexity of CLEAR-PGD.

\subsection{RQ3: Which nodes are altered and how?}

To understand how CLEAR-PGD achieves the best conflict reduction, we examine the shifts $\Delta p_i = p_i - p_i^{(0)}$, $\Delta q_i = q_i - q_i^{(0)}$,  and node investment $\tau_i = |\Delta p_i| + |\Delta q_i|$, against opinion extremity $e_i = |s_i - 0.5|$, graph positioning ($\pi_i$), and activity behavior of the nodes. 
The Pearson correlations and mean shifts across all datasets, for $b/n = 0.4$, are reported in Table~\ref{tab:appendix_rq3} of Appendix~\ref{app:Experiments}. 
We identify four intervention strategies applied by CLEAR-PGD:

\begin{itemize}[leftmargin=*]
    \item \textbf{Exposing Extremists:} Ideological extremity $e_i = |s_i - 0.5|$ exhibits a strong positive correlation with observing shifts $\Delta q_i$ ($r \in [+0.35, +0.51]$ in Scenario A; $r \in [+0.60, +0.72]$ in Scenario B). Algorithm~\ref{alg:clear_pgd} increases $q_i$ for extreme users, compelling extreme-minded nodes to listen to surrounding opinions.
    \item \textbf{Throttling Extremists and Amplifying Moderates:} Conversely, ideological extremity correlates negatively with posting shifts $\Delta p_i$ across every graph ($r \in [-0.24, -0.59]$ in Scenario A; $r \in [-0.22, -0.35]$ in Scenario B). This indicates that CLEAR-PGD reduces amplification for ideological extremists while boosting posting probabilities ($\Delta p_i > 0$) for moderate users.
    \item \textbf{Central Node Targeting and Activity Rebalancing:} Total node investment $\tau_i$ correlates positively with PageRank centrality $\pi_i$ ($r \in [+0.08, +0.52]$), showing that CLEAR-PGD prioritizes central hubs. The top 1\% broadcasters undergo large increases in $\mathbf{q}$, transforming central influencers into active listeners. Simultaneously, contributor (9\%) and observer (90\%) nodes receive substantial amplification boosts ($\overline{\Delta p} \in [+0.19, +0.38]$), empowering the silent majority of moderate nodes to participate.
    \item \textbf{Budget Allocation ($b$):} When jointly optimizing vectors $\mathbf{p}$ and $\mathbf{q}$, CLEAR-PGD consistently allocates $61\%\text{--}72\%$ of the total budget $b$ to modifying posting probabilities and $28\%\text{--}39\%$ to modifying observing probabilities. This confirms our RQ2 observation that amplification ($\mathbf{p}$) serves as the primary driver of conflict reduction, while exposure ($\mathbf{q}$) acts as a complementary lever.
\end{itemize}


\subsection{RQ4: How does CLEAR-PGD scale?}

Figure~\ref{fig:rq3_scalability} in Appendix \ref{app:Experiments} shows the execution time of CLEAR-PGD across networks of increasing size under Scenario B ($b/n = 0.2$), plotted against graph edges on a log-log scale.
We report the mean over five random initializations, with error bars of one standard deviation.
For smaller networks (Gnutella and Wiki), execution times plateau between $1$ and $2$ seconds since the runtime is dominated by constant-time hardware overhead.
For larger networks ranging from Epinions to Pokec, the empirical trend indicates a polynomial scaling relationship proportional to $|E|^{\gamma}$, with the estimated exponent revealing sublinear growth ($\gamma \approx0.8 < 1$).



\section{Conclusion}
\label{sec:conclusion}

In this paper, we introduced the PO-FJ model, extending classic opinion dynamics to capture diverse user activity behaviors on social media. 
We showed that under PO-FJ dynamics, opinions are driven into perpetual oscillations, generalizing established opinion-formation metrics. 
Using the new levers of PO-FJ, we formulated the CLEAR problem for social-network conflict mitigation and developed CLEAR-PGD, a scalable algorithm that effectively reduces conflict in large-scale networks. 
These results open new research directions for modeling complex web interactions. 
Future work can investigate coordinated group behavior via the activity covariances $\Sigma^x$ and $\Sigma^y$, relax the independent-activity assumption, and explore alternative objectives, such as maximizing network fairness.




\bibliographystyle{ACM-Reference-Format}
\bibliography{robustness-bibliography}


\appendix

\section{Proof of Theorem \ref{thm:no_static_consensus}}
\label{app:no_static_consensus}

We prove the impossibility of converging to a static equilibrium and the convergence in distribution to a unique, non-degenerate (i.e., not concentrated at a single point) stationary probability distribution.

\spara{Part 1: Impossibility of Static Equilibrium.} 
Assume by contradiction that $\hat{\mathbf{z}}^{(t)}$ converges almost surely to a static equilibrium vector $\mathbf{z}^*$ as $t \to \infty$, or that the stationary distribution is a degenerate point mass concentrated at $\mathbf{z}^*$. 
By continuity of the affine update in Equation~\eqref{eq:PO_FJ_update}, the fixed-point equation $\mathbf{z}^* = A \mathbf{s} + (I - A) \hat{W}^{(t)} \mathbf{z}^*$ must hold with probability $1$ for every realization $\hat{W}^{(t)}$ in the support of the activity distribution, and taking the expectation yields $\mathbf{z}^* = A \mathbf{s} + (I - A) \overline{W} \mathbf{z}^*$.

Consider any node $i \in \{1, \dots, n\}$. Because $q_i = \mathbb{E}[x_i^{(t)}] \in (0,1)$, the marginal probability that node $i$ does not observe at iteration $t$ is strictly positive: $\Pr(x_i^{(t)} = 0) = 1 - q_i > 0$, which holds under Assumption~\ref{ass:activity} for any valid covariance matrices $\Sigma^x, \Sigma^y$. 
By Equation~\eqref{eq:Dynamic Interaction Matrix}, whenever $x_i^{(t)} = 0$, the $i$-th row of $\hat{W}^{(t)}$ collapses to the standard basis vector $\mathbf{e}_i^T$ ($\hat{W}_{ii}^{(t)} = 1$ and $\hat{W}_{ij}^{(t)} = 0$ for all $j \neq i$). 
Evaluating the $i$-th coordinate of $\mathbf{z}^* = A \mathbf{s} + (I - A) \hat{W}^{(t)} \mathbf{z}^*$ on the positive-probability event $\{x_i^{(t)} = 0\}$ yields:
\begin{equation}
    z_i^* = a_i s_i + (1 - a_i) z_i^* \implies a_i (z_i^* - s_i) = 0.
\end{equation}
Since $a_i \in (0,1)$ for all nodes, this forces $z_i^* = s_i$ for every $i \in \{1, \dots, n\}$, and hence $\mathbf{z}^* = \mathbf{s}$.

Substituting $\mathbf{z}^* = \mathbf{s}$ into the expected equilibrium relation $\mathbf{z}^* = A \mathbf{s} + (I - A) \overline{W} \mathbf{z}^*$ gives $\mathbf{s} = A \mathbf{s} + (I - A) \overline{W} \mathbf{s}$, which rearranges to $(I - A)(I - \overline{W})\mathbf{s} = \mathbf{0}$. 
Because $(I - A)$ is diagonal with strictly positive entries $1 - a_i > 0$, it is invertible, implying $\overline{W}\mathbf{s} = \mathbf{s}$. 
Since $W$ is irreducible and $p_i, q_i \in (0,1)$, the expected interaction matrix $\overline{W}$ is irreducible and row-stochastic (Definition~\ref{def:expected_W}). By the Perron--Frobenius theorem, the eigenspace of $\overline{W}$ associated with eigenvalue $1$ is one-dimensional and spanned by $\mathbf{1}$, forcing $\mathbf{s} = c\mathbf{1}$ for some constant $c \in \mathbb{R}$. 
This contradicts the assumption that inner opinions are not identical ($s_i \neq s_j$ for some $i, j$), proving that $\{\hat{\mathbf{z}}^{(t)}\}_{t\ge 0}$ cannot converge to a static point and that no stationary distribution can be concentrated at a single point.

\spara{Part 2: Convergence to a Unique Stationary Distribution.}
While the network does not reach a static consensus, the opinions do not diverge. The opinion update rule is of the form $\hat{\mathbf{z}}^{(t+1)} = M_t \hat{\mathbf{z}}^{(t)} + \mathbf{b}$,
where $M_t = (I - A) \hat{W}^{(t)}$ is a random $n \times n$ matrix and $\mathbf{b} = A \mathbf{s}$ is a constant $n \times 1$ vector. By Assumption~\ref{ass:activity}, matrices $M_t$ are independent and identically distributed over time.
To analyze stability, we evaluate how matrix $M_t$ scales opinions under the infinity norm. Every realization of the dynamic interaction matrix $\hat{W}^{(t)}$ is row-stochastic, meaning $\|\hat{W}^{(t)}\|_\infty = 1$. 
Because all agents have stubbornness $a_i \in (0,1)$, the diagonal matrix $(I-A)$ has entries $(1-a_i) \in (0, 1)$. Therefore, for every realization at iteration $t$, the norm of $M_t$ is strictly bounded by a constant $\gamma$:
\begin{equation*}
    \| M_t \|_\infty = \| (I - A) \hat{W}^{(t)} \|_\infty \le \max_i (1 - a_i) \eqqcolon \gamma < 1.
\end{equation*}
Inequality $\gamma < 1$ establishes that every realization $M_t$ is a strict contraction mapping.
By the foundational theorem of Iterated Random Functions (Theorem 2.1 of~\cite{Diaconis}), the affine recursion admits a unique invariant probability distribution if the expected logarithmic Lipschitz norm is strictly negative.
Because $\|M_t\|_\infty \le \gamma < 1$ for all $t$, we have $\frac{1}{t}\mathbb{E}[\log \|M_1 \dots M_t\|_\infty] \le \log(\gamma) < 0$, satisfying this condition.
Thus, $\{\hat{\mathbf{z}}^{(t)}\}_{t\ge0}$ converges in distribution to a unique stationary probability distribution, which by Part 1 is non-degenerate with strictly positive variance. \qed

\section{Proof of Lemma \ref{lm:expected_opinions}}
\label{app:Lemma1}

To derive the evolution of $\boldsymbol{\mu}^{(t)} = \mathbb{E}[\mathbf{\hat{z}}^{(t)}]$, we apply the expectation operator to both sides of Eq. \eqref{eq:PO_FJ_update}. 
Because $A$ and $\mathbf{s}$ are deterministic constants, this yields: $\boldsymbol{\mu}^{(t+1)} = A \mathbf{s} + (I - A) \mathbb{E}\big[\hat{W}^{(t)} \mathbf{\hat{z}}^{(t)}\big]$.
Under the PO-FJ Assumption \ref{ass:activity}, interaction matrix $\hat{W}^{(t)}$ is independent of $\mathbf{\hat{z}}^{(t)}$, since it relies on the interactions of up to time $t-1$. 
This independence allows us to decompose the expectation: $\mathbb{E}\big[\hat{W}^{(t)} \mathbf{\hat{z}}^{(t)}\big] = \mathbb{E}\big[\hat{W}^{(t)}\big] \mathbb{E}\big[\mathbf{\hat{z}}^{(t)}\big] = \overline{W} \boldsymbol{\mu}^{(t)}$.
Substituting this yields Eq. \eqref{eq:meanrec}.
Notice that this is the exact form of classic FJ dynamics. 
Because $\overline{W}$ retains the row-stochasticity and connectivity of $W$, this sequence inherits standard FJ convergence. Thus, this sequence converges to the unique fixed point of Eq. \eqref{eq:meanss}.

\section{Proof of Lemma \ref{lm:covariance}}
\label{app:Lemma2}

To derive the exact update equation for the covariance matrix $M^{(t)}$, we analyze the evolution of the opinion \emph{error vector}, defined as the deviation of the network from its expected state: $\mathbf{e}^{(t)} = \mathbf{\hat{z}}^{(t)} - \boldsymbol{\mu}^{(t)}$. 
By definition, $\mathbb{E}[\mathbf{e}^{(t)}] = \mathbf{0}$, and the covariance matrix is: $M^{(t)} = \mathbb{E}\big[\mathbf{e}^{(t)} (\mathbf{e}^{(t)})^T\big]$.
Subtracting the expected update $\boldsymbol{\mu}^{(t)}$ (Lemma \ref{lm:expected_opinions}) from Equation \eqref{eq:PO_FJ_update}, the stubbornness term $A\mathbf{s}$ perfectly cancels out:
\begin{equation} \label{eq:error_step1}
    \mathbf{e}^{(t+1)} = (I-A) \big[ \hat{W}^{(t)} \mathbf{\hat{z}}^{(t)} - \overline{W}\boldsymbol{\mu}^{(t)} \big].
\end{equation}
We apply centering by adding and subtracting $\hat{W}^{(t)}\boldsymbol{\mu}^{(t)}$ inside the bracket, which allows us to isolate the error vector $\mathbf{e}^{(t)}$:
\begin{align}
    \mathbf{e}^{(t+1)} 
    &= (I-A) \big[ \hat{W}^{(t)}\mathbf{e}^{(t)} + \Delta W^{(t)} \boldsymbol{\mu}^{(t)} \big], \label{eq:error_step2}
\end{align}
where $\Delta W^{(t)} = \hat{W}^{(t)} - \overline{W}$.
To find the new covariance matrix $M^{(t+1)}$, we compute the expectation of the outer product $\mathbb{E}\big[\mathbf{e}^{(t+1)} (\mathbf{e}^{(t+1)})^T\big]$. Expanding Equation \eqref{eq:error_step2} yields two quadratic terms and a cross-term. 
Crucially, because the network sampling at time $t$ is independent of the network's history, the matrices $\hat{W}^{(t)}$ and $\Delta W^{(t)}$ are strictly independent of the current error $\mathbf{e}^{(t)}$. 
Because $\mathbb{E}[\mathbf{e}^{(t)}] = \mathbf{0}$, it holds that: $\mathbb{E}\big[ \hat{W}^{(t)}\mathbf{e}^{(t)} (\boldsymbol{\mu}^{(t)})^T (\Delta W^{(t)})^T \big]= \mathbf{0}$.
Therefore, the covariance update decouples into two components:
\begin{align} 
    M^{(t+1)} &= (I-A) \mathbb{E}\big[ \hat{W}^{(t)} M^{(t)} (\hat{W}^{(t)})^T \big] (I-A) \nonumber \\
    &\quad + (I-A) \mathbb{E}\big[ \Delta W^{(t)} \boldsymbol{\mu}^{(t)} (\boldsymbol{\mu}^{(t)})^T (\Delta W^{(t)})^T \big] (I-A) \nonumber \\
    &= F(M^{(t)}) + \Omega^{(t)}. \label{eq:cov_decoupled}
\end{align}

\spara{Deriving $F(M^{(t)})$:} Applying the variance of a product to the independent random activity vectors $\mathbf{x}^{(t)}$ and $\mathbf{y}^{(t)}$, the expectation expands into the mean propagation $\overline{W} M^{(t)} \overline{W}^T$ plus the variance operator $\mathcal{H}(M^{(t)})$. 
The function $\mathcal{G}$ emerges from computing the Hadamard products ($\circ$) of the covariance matrices ($\Sigma^x, \Sigma^y$) across the base interaction matrix $W$.

\spara{Deriving $\Omega^{(t)}$:} Expanding the expectation of the zero-mean noise $\Delta W^{(t)}$ over $\boldsymbol{\mu}^{(t)}$ yields operator: $B_{\boldsymbol{\mu}^{(t)}} = W\text{diag}(\boldsymbol{\mu}^{(t)}) - \text{diag}(\boldsymbol{\mu}^{(t)})W$. 
Scaling this operator by the activity covariance matrices $\Sigma^x$ and $\Sigma^y$ directly produces the closed-form Hadamard formulation of the variance injection $\Omega^{(t)}$ defined in Lemma \ref{lm:covariance}. \qed

\section{The Discrete-CLEAR Problem}
\label{app:nphard}

In Section~\ref{sec:mitigation}, we introduced the CLEAR problem (Problem~\ref{prob:clear}), where platform parameter control is $\mathbf{p}, \mathbf{q} \in [0,1]^n$. 
Here, we define its combinatorial counterpart, where a platform completely mutes or activates a subset of at most $k$ users, and prove that it is NP-hard.

\begin{problem}[Discrete-CLEAR]\label{prob:discrete_clear}
Given a PO-FJ network with initial probability vectors $\mathbf{p}^{(0)}, \mathbf{q}^{(0)} \in [0,1]^n$ and an integer intervention budget $k < n$, select a subset of at most $k$ nodes and assign for each selected node $p_i, q_i \in \{0,1\}$ such that $\overline{\mathcal{C}}(\mathbf{p}, \mathbf{q})$ is minimized.
\end{problem}

\begin{proof}
To prove that the Discrete-CLEAR problem is NP-hard, we construct a poly-time reduction from the Vertex Cover problem. 

\spara{Decision Discrete-CLEAR Problem:} Given a PO-FJ network, an intervention budget $k$, and a target conflict threshold $C^*$, does there exist a set of at most $k$ node interventions ($p_i, q_i \in \{0,1\}$) such that the total network conflict satisfies $\overline{\mathcal{C}}(\mathbf{p}, \mathbf{q}) \le C^*$?

\spara{Decision Vertex Cover Problem:} Given an undirected graph $G_v = (V_v, E_v)$ and an integer $k$, does there exist a subset of vertices $S \subseteq V_v$ with $|S| \le k$ such that every edge in $E_v$ has at least one endpoint in $S$?

\spara{Construction.} Let $G_v$ be an arbitrary instance of the Vertex Cover problem with $n$ nodes.
Assume w.l.o.g.\ that $G_v$ has no isolated vertices.
We construct an instance of the PO-FJ network in polynomial time as follows: 
Let the nodes be $V = V_v$ and the edges be $E = E_v$. 
Set the elements of the base interaction matrix $W$ as $W_{ij} = 1/d_i$ for all $(i,j) \in E$, where $d_i$ is the degree of node $i$. 
Assign the inner opinions as $s_i = i/n$ for all $i \in V$.
Set the initial posting and observing probabilities as $p^{(0)}_i = q^{(0)}_i = 1/2$.
Set the stubbornness for all nodes as $a_i = 1 - \epsilon$, with $\Sigma^x, \Sigma^y$ set to any valid covariance matrices (e.g., diagonal). We set:
\[
  \epsilon=\frac{1}{100\,n^{4}},\qquad
  C^*=\mathcal{P}({\mathbf{s}})=\frac{1}{n}\sum_{i=1}^n (s_i-\bar{s})^2,\quad \bar{s}=\frac{1}{n}\sum_{j=1}^n s_j .
\]

Observe the following two facts for the above assignment: 
\emph{(F1)} $\overline{W}$ is row-stochastic, so
$\overline{\boldsymbol{\mu}}-\mathbf{s}=(1-\epsilon)\sum_{t\ge1}\epsilon^{t}(\overline{W}^{t}\mathbf{s}-\mathbf{s})$, and thus
$\|\overline{\boldsymbol{\mu}}-\mathbf{s}\|_\infty\le\epsilon$.
\emph{(F2)} The $\overline{M}$-terms for Polarization and Disagreement in Equations~\eqref{eq:pol_POFJ_stas} and~\eqref{eq:dis_POFJ_stas} are nonnegative. Using $\mathbf{e}=\mathbf{\hat{z}}^{(t)}-\overline{\boldsymbol{\mu}}$ and $\bar{e}=\frac{1}{n}\sum_i e_i$, these terms can be written at stationarity as:
\begin{align*}
 \frac{1}{n}\mathrm{Tr}\!\big(\overline{M}(I-\tfrac{1}{n}\mathbf{1}\mathbf{1}^{\top})\big) 
 &=\mathbb{E}\Big[\frac{1}{n}\sum_{i=1}^n (e_i-\bar{e})^2\Big]\ge 0,\\
 \frac{1}{2n}\sum_{i,j}\overline{W}_{ij}(\overline{M}_{ii}+\overline{M}_{jj}-2\overline{M}_{ij}) 
 &=\frac{1}{2n}\sum_{i,j}\overline{W}_{ij}\,\mathrm{Var}(\hat{z}_i^{(t)}-\hat{z}_j^{(t)})\ge 0.
\end{align*}

\spara{Forward ($\Rightarrow$).} Let $S$ be a vertex cover with $|S|\le k$, and set $(p_i,q_i)=(0,0)$ for all $i\in S$.
Every edge has at least one endpoint in $S$, so $\widehat{W}^{(t)}_{ij}=x_i^{(t)}y_j^{(t)}W_{ij}=0$ for all $i\ne j$.
From the PO-FJ dynamics, this results in $\widehat{W}^{(t)}=I, \forall t$. 
Consequently, $\mathbf{\hat{z}}^{(t)}\to\mathbf{s}$ deterministically, yielding $\overline{\boldsymbol{\mu}}=\mathbf{s}$,
$\overline{M}=0$, $\overline{\mathcal{D}}=0$, and $\overline{\mathcal{C}}=\mathcal{P}({\mathbf{s}})=C^*$.

\spara{Reverse ($\Leftarrow$).} Let $T'$ be the set of intervened nodes with $|T'|\le k$, and suppose
$T'$ is not a vertex cover. Then some edge $\{u,v\} \in E_v$ has $u,v\notin T'$, so
$\overline{W}_{uv}=q_u p_v W_{uv}=W_{uv}/4\ge 1/(4n)$, while $|s_u-s_v|\ge 1/n$.
By (F1) and $\epsilon\le 1/(4n)$, we have $|\overline{\mu}_u-\overline{\mu}_v|\ge 1/(2n)$. Dropping the nonnegative
$\overline{M}$-terms (F2) and all other nonnegative summands, we obtain:
\[
 \overline{\mathcal{D}}\ \ge\ \frac{1}{2n} \overline{W}_{uv} (\overline{\mu}_u - \overline{\mu}_v)^2 \ge\ \frac{1}{2n}\cdot\frac{1}{4n}\cdot\frac{1}{4n^{2}}=\frac{1}{32n^{4}} .
\]
Writing $\overline{\boldsymbol{\mu}}=\mathbf{s}+\boldsymbol{\delta}$ with $\|\boldsymbol{\delta}\|_\infty\le\epsilon$, and noting that $|s_i-\bar{s}| < 1$ and $\sum_i (s_i - \bar{s}) = 0$, we expand the polarization:
\[
\overline{\mathcal{P}} \geq \mathcal{P}(\overline{\boldsymbol{\mu}})=\mathcal{P}(\mathbf{s})+\frac{2}{n}\sum_{i=1}^n(s_i-\bar{s})\delta_i+\mathcal{P}(\boldsymbol{\delta})\ge\mathcal{P}(\mathbf{s})-2\epsilon.
\]
Therefore, total conflict is strictly bounded by: 
\[
 \overline{\mathcal{C}} = \overline{\mathcal{P}} + \overline{\mathcal{D}} \ge\ \mathcal{P}(\mathbf{s})-2\epsilon+\frac{1}{32n^{4}}>\mathcal{P}(\mathbf{s})=C^*,
\]
(since $\epsilon < 1/(64n^4)$), which is a contradiction. Hence, $T'$ must be a valid vertex cover of size at most $k$.
\end{proof}

\section{Updated Metrics \eqref{eq:po_POFJ_stas} \eqref{eq:pol_POFJ_stas} \eqref{eq:dis_POFJ_stas}}
\label{app:PolDis}

To derive the stationary expected public opinion, polarization, and disagreement, we take the limit and expectation of the Equations \eqref{eq:po_POFJ}, \eqref{eq:pol_POFJ}, \eqref{eq:dis_POFJ}: 

\spara{Stationary Expected Public Opinion $\overline{\mathcal{Z}}$:}
By definition, the public opinion at iteration $t$ is $\mathcal{Z}^{(t)} = \frac{1}{n} \sum_{i=1}^n \hat{z}_i^{(t)}$. Taking the expectation and applying linearity of expectation yields:
$$
\mathbb{E}\left[\mathcal{Z}^{(t)}\right] = \frac{1}{n} \sum_{i=1}^n \mathbb{E}\left[\hat{z}_i^{(t)}\right] = \frac{1}{n} \sum_{i=1}^n \mu_i^{(t)}.
$$
Taking the limit as $t \to \infty$ and substituting Lemma~\ref{lm:expected_opinions} ($\lim_{t \to \infty} \boldsymbol{\mu}^{(t)} = \overline{\boldsymbol{\mu}}$), we obtain the stationary expected public opinion of Eq. \eqref{eq:po_POFJ_stas}.

\spara{Stationary Expected Polarization $\overline{\mathcal{P}}$:}
We define the symmetric matrix $C = I - \frac{1}{n}\mathbf{1}\mathbf{1}^\top$. This allows us to write the PO-FJ iteration-based polarization as $\mathcal{P}^{(t)} = \frac{1}{n} (\mathbf{\hat{z}}^{(t)})^\top C \mathbf{\hat{z}}^{(t)}$. By applying the identity $\mathbb{E}[\mathbf{x}^\top B \mathbf{x}] = \mathbb{E}[\mathbf{x}]^\top B \mathbb{E}[\mathbf{x}] + \text{Tr}(B \,\text{Cov}(\mathbf{x}))$ and limit, we obtain:
$$
\overline{\mathcal{P}} = \lim_{t \to \infty} \mathbb{E}\left[\mathcal{P}^{(t)}\right] = \frac{1}{n} \overline{\boldsymbol{\mu}}^\top C \overline{\boldsymbol{\mu}} + \frac{1}{n} \text{Tr}\left(C \overline{M}\right).
$$
We expand this expression into its scalar components. The first term yields the variance of the stationary expected opinions $\overline{\boldsymbol{\mu}}$. The trace term expands to $\frac{1}{n}\text{Tr}(\overline{M}) - \frac{1}{n^2}\mathbf{1}^\top \overline{M} \mathbf{1}$. Combining these expansions yields the final closed-form expression of Eq. \eqref{eq:pol_POFJ_stas}.

\spara{Stationary Expected Disagreement $\overline{\mathcal{D}}$:}
The disagreement metric relies on the dynamic interaction matrix $\hat{W}^{(t)}$. Crucially, while $\mathbf{\hat{z}}^{(t)}$ is determined by the history of interactions up to $t-1$, $\hat{W}^{(t)}$ depends solely on the random posting and observing activity at iteration $t$. Because they are strictly independent, we can write:
\begin{align*}
\mathbb{E}\left[ \mathcal{D}^{(t)} \right] & = \frac{1}{2n} \sum_{i=1}^n \sum_{j=1}^n \mathbb{E}\left[\hat{W}_{ij}^{(t)}\right] \mathbb{E}\left[\left(\hat{z}_i^{(t)} - \hat{z}_j^{(t)}\right)^2\right].
\end{align*}
This maps the disagreement directly to the expected interaction matrix $\overline{W}$. To evaluate this, we define $L_{\overline{W}} = \frac{1}{2}(I + D_{\text{in}}(\overline{W}) - \overline{W} - \overline{W}^\top)$, where $D_{\text{in}}$ is the diagonal in-degree matrix, noting that $D_{\text{out}}(\overline{W}) = I$ due to row-stochasticity. 
Applying these results yields:
$$
\overline{\mathcal{D}} = \lim_{t \to \infty} \mathbb{E}\left[ \frac{1}{n} (\mathbf{\hat{z}}^{(t)})^\top L_{\overline{W}} \mathbf{\hat{z}}^{(t)} \right] = \frac{1}{n} \overline{\boldsymbol{\mu}}^\top L_{\overline{W}} \overline{\boldsymbol{\mu}} + \frac{1}{n} \text{Tr}\left(L_{\overline{W}} \overline{M}\right).
$$
The first term represents the static expected disagreement. Similarly to before, the trace term expands and also incorporates $\overline{W}_{ij}$ and $(1 + \overline{d}_i^{\text{in}})$, due to $L_{\overline{W}}$, where $\overline{d}_i^{\text{in}} = \sum_k \overline{W}_{ki}$, resulting in Eq. \eqref{eq:dis_POFJ_stas}.

\section{$K$-hop Matrix $\overline{M}$ Approximation}
\label{app:khop_approx}

To evaluate $\overline{M}$ by iterating Equation~\eqref{eq:covrec}, the resulting matrix is \emph{dense}: on a strongly connected graph, every pair of nodes eventually acquires a nonzero covariance ($\overline{M}_{ij} \neq 0$). 
Storing and updating a dense matrix is computationally prohibitive (e.g., $n = 10^6$ requires $8$~TB of memory). 
However, in sparse social networks, the off-diagonal entries $\overline{M}_{ij}$ decay geometrically because distant nodes exert weak effects on one another. 
Consequently, $M^{(t)}$ and $\overline{M}$ can be accurately approximated via a $K$-hop structural truncation that zeros out covariance entries between nodes that are beyond $K$ hops apart in each iteration of~\eqref{eq:covrec}. 
For $K=1$ and for diagonal $\Sigma^x, \Sigma^y$ using sparse matrix-vector products and edge-indexed operations results in $\mathcal{O}(|E|)$ complexity per iteration of Equation~\eqref{eq:covrec}.

\spara{Empirical Validation.}
To quantify the accuracy of the structural truncation, we evaluate both how much matrix information the sparse approximation retains and how closely it preserves the true network conflict $\overline{\mathcal{C}}$. 
A standard metric for retained matrix information~\cite{golub2013matrix} is the \emph{Retained Energy Ratio} ($\rho_K$), calculated via the Frobenius norm, alongside the \emph{Relative Conflict Error} ($\epsilon_{\mathcal{C}}$):
\begin{equation}
    \rho_K = \frac{\|\overline{M}^{(K)}\|_F^2}{\|\overline{M}^{\text{dense}}\|_F^2}, \qquad \epsilon_{\mathcal{C}} = \frac{|\overline{\mathcal{C}}^{(K)} - \overline{\mathcal{C}}^{\text{dense}}|}{\overline{\mathcal{C}}^{\text{dense}}},
\end{equation}
where $\overline{M}^{\text{dense}}$ and $\overline{\mathcal{C}}^{\text{dense}}$ denote the exact dense stationary opinion covariance matrix calculated from Eq.~\eqref{eq:covrec} and the resulting conflict, while $\overline{M}^{(K)}$ and $\overline{\mathcal{C}}^{(K)}$ denote the iteratively truncated $K$-hop approximation counterparts.
We evaluate these metrics for $K \in \{1, 2, 3\}$ on the Wiki dataset and a synthetic Stochastic Block Model (SBM) graph with $n = 2{,}000$ nodes partitioned into two equal-sized communities with edge probabilities set to $0.06$ for intra-community and $0.005$ for inter-community connections.

\begin{table}[h!]
    \centering
    \caption{Analysis of the $K$-hop iterative approximation results: sparsity, Frobenius energy retention $\rho_K$, and relative conflict error $\epsilon_{\mathcal{C}}$.}
    \label{tab:k_hop_analysis}
    \resizebox{\columnwidth}{!}{%
    \begin{tabular}{l c c c c c c}
        \toprule
        \multirow{2}{*}{\textbf{Approximation}} & \multicolumn{3}{c}{\textbf{Wiki}} & \multicolumn{3}{c}{\textbf{SBM}} \\
        \cmidrule(lr){2-4} \cmidrule(lr){5-7}
        & \textbf{Sparsity ($|E_K| / n^2$)} & \textbf{$\rho_K$} & \textbf{$\epsilon_{\mathcal{C}}$} & \textbf{Sparsity ($|E_K| / n^2$)} & \textbf{$\rho_K$} & \textbf{$\epsilon_{\mathcal{C}}$} \\
        \midrule
        $K=1$ (1-Hop) & 2.41\% & 0.9061 & 0.0023\% & 3.30\% & 0.9970 & 0.0001\% \\
        $K=2$ (2-Hop) & 33.02\% & 0.9329 & 0.0012\% & 71.38\% & 0.9994 & 0.0000\% \\
        $K=3$ (3-Hop) & 80.52\% & 0.9702 & 0.0008\% & 100.00\% & 1.0000 & 0.0000\% \\
        Dense (Exact) & 100.00\% & 1.0000 & 0.0000\% & 100.00\% & 1.0000 & 0.0000\% \\
        \bottomrule
    \end{tabular}%
    }
\end{table}

As shown in Table~\ref{tab:k_hop_analysis}, the $1$-hop iterative approximation ($K=1$) captures over $90.6\%$ of the Frobenius energy on Wiki and $99.7\%$ on SBM while yielding virtually zero error ($\epsilon_{\mathcal{C}} < 0.01\%$) in the conflict objective. 
While $K=2$ and $K=3$ provide marginal energy gains, they do not justify the steep increase in computational and memory cost (density jumping from $2.41\%$ to $33.02\%$ on Wiki and from $3.30\%$ to $71.38\%$ on SBM). 
Therefore, applying the structural approximation with $K=1$ anchors the per-iteration complexity to $\mathcal{O}(|E|)$, enabling scaling to millions of nodes while preserving most of the retained energy of $\overline{M}$ and keeping the conflict objective $\overline{\mathcal{C}}$ very close to the actual value.

\section{Extended Experiment Results}
\label{app:Experiments}

Table~\ref{tab:appendix_rq2} provides the results supporting RQ2, showing the conflict across the remaining networks of Table~\ref{tab:datasets_main} at budgets ($b/n \in \{0,0.4, 1.0\}$).
Table~\ref{tab:appendix_rq3} provides the results supporting RQ3 across all datasets and scenarios for $b/n = 0.4$. 
We report the Pearson correlations ($r$) of the shifts $(\Delta p_i, \Delta q_i)$ with $e_i = |s_i - 0.5|$, $\tau_i = |\Delta p_i| + |\Delta q_i|$ and PageRank $\pi_i$. Also we show mean shifts across initial activity levels, and the budget allocation split between $\mathbf{p}$ and $\mathbf{q}$.

\begin{figure}[b]
    \centering
    \vspace{-6pt}
    \includegraphics[width=0.85\columnwidth]{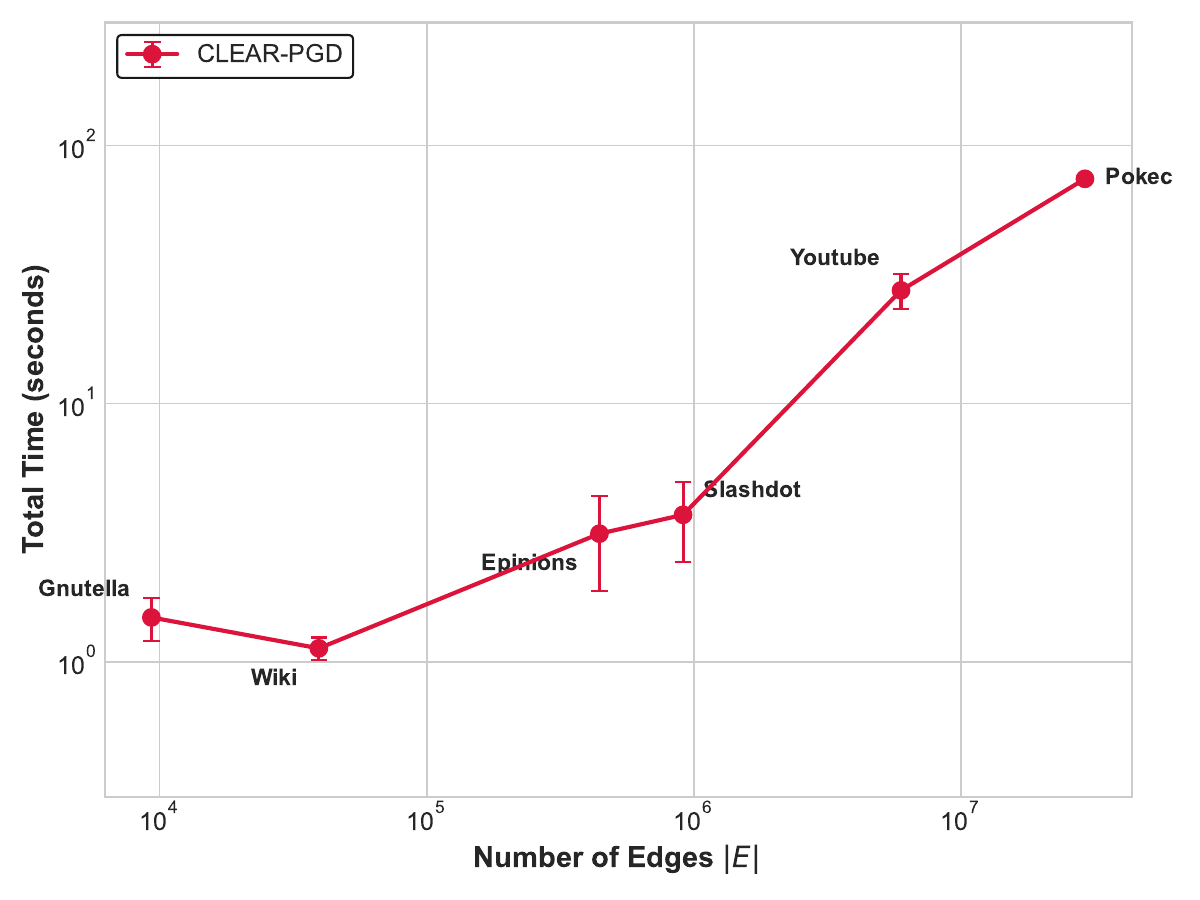}
    \vspace{-6pt}
    \caption{Total time required to solve the CLEAR Problem.}
    \label{fig:rq3_scalability}
    \vspace{-6pt}
\end{figure}

\begin{table}[h!]
\centering
\caption{Conflict results ($\overline{\mathcal{C}}$) for additional networks.}
\label{tab:appendix_rq2}
\vspace{-6pt}
\resizebox{\columnwidth}{!}{%
\begin{tabular}{llccccccc}
\toprule
\textbf{Dataset} & \textbf{Method} & \multicolumn{3}{c}{\textbf{Scenario A ($\overline{\mathcal{C}}$)}} & & \multicolumn{3}{c}{\textbf{Scenario B ($\overline{\mathcal{C}}$)}} \\
\cmidrule{3-5} \cmidrule{7-9}
 & & $b/n=0$ & $b/n=0.4$ & $b/n=1.0$ & & $b/n=0$ & $b/n=0.4$ & $b/n=1.0$ \\
\midrule
\multirow{6}{*}{\textbf{Wiki}} 
 & Random                 & 0.1227 & 0.1194 & 0.1082 & & 0.0695 & 0.0631 & 0.0501 \\
 & Ideological-Heuristic  & 0.1227 & 0.1136 & 0.1096 & & 0.0695 & 0.0593 & 0.0540 \\
 & Bridge-Heuristic       & 0.1227 & 0.1211 & 0.1200 & & 0.0695 & 0.0677 & 0.0666 \\
 & CLEAR-PGD ($q$-only)   & 0.1227 & 0.1196 & 0.1177 & & 0.0695 & 0.0667 & 0.0651 \\
 & CLEAR-PGD ($p$-only)   & 0.1227 & 0.1105 & 0.0996 & & 0.0695 & 0.0566 & 0.0449 \\
 & CLEAR-PGD              & 0.1227 & \textbf{0.1069} & \textbf{0.0922} & & 0.0695 & \textbf{0.0550} & \textbf{0.0401} \\
\midrule
\multirow{6}{*}{\textbf{Gnutella}} 
 & Random                 & 0.1151 & 0.1079 & 0.0917 & & 0.0703 & 0.0638 & 0.0511 \\
 & Ideological-Heuristic  & 0.1151 & 0.1074 & 0.1077 & & 0.0703 & 0.0592 & 0.0564 \\
 & Bridge-Heuristic       & 0.1151 & 0.1120 & 0.1102 & & 0.0703 & 0.0677 & 0.0663 \\
 & CLEAR-PGD ($q$-only)   & 0.1151 & 0.1105 & 0.1091 & & 0.0703 & 0.0660 & 0.0646 \\
 & CLEAR-PGD ($p$-only)   & 0.1151 & 0.0968 & 0.0846 & & 0.0703 & 0.0561 & 0.0452 \\
 & CLEAR-PGD              & 0.1151 & \textbf{0.0934} & \textbf{0.0756} & & 0.0703 & \textbf{0.0536} & \textbf{0.0387} \\
\midrule
\multirow{6}{*}{\textbf{Slashdot}} 
 & Random                 & 0.1011 & 0.0964 & 0.0842 & & 0.0687 & 0.0650 & 0.0548 \\
 & Ideological-Heuristic  & 0.1011 & 0.1021 & 0.1070 & & 0.0687 & 0.0654 & 0.0619 \\
 & Bridge-Heuristic       & 0.1011 & 0.0980 & 0.0973 & & 0.0687 & 0.0644 & 0.0637 \\
 & CLEAR-PGD ($q$-only)   & 0.1011 & 0.0963 & 0.0958 & & 0.0687 & 0.0626 & 0.0620 \\
 & CLEAR-PGD ($p$-only)   & 0.1011 & 0.0897 & 0.0821 & & 0.0687 & 0.0595 & 0.0540 \\
 & CLEAR-PGD              & 0.1011 & \textbf{0.0842} & \textbf{0.0748} & & 0.0687 & \textbf{0.0551} & \textbf{0.0458} \\
\midrule
\multirow{6}{*}{\textbf{YouTube}} 
 & Random                 & 0.1084 & 0.1023 & 0.0912 & & 0.0655 & 0.0617 & 0.0511 \\
 & Ideological-Heuristic  & 0.1084 & 0.1063 & 0.1071 & & 0.0655 & 0.0635 & 0.0594 \\
 & Bridge-Heuristic       & 0.1084 & 0.1042 & 0.1035 & & 0.0655 & 0.0597 & 0.0591 \\
 & CLEAR-PGD ($q$-only)   & 0.1084 & 0.1030 & 0.1016 & & 0.0655 & 0.0563 & 0.0554 \\
 & CLEAR-PGD ($p$-only)   & 0.1084 & 0.0951 & 0.0863 & & 0.0655 & 0.0544 & 0.0503 \\
 & CLEAR-PGD              & 0.1084 & \textbf{0.0894} & \textbf{0.0787} & & 0.0655 & \textbf{0.0462} & \textbf{0.0393} \\
\midrule
\multirow{6}{*}{\textbf{Pokec}} 
 & Random                 & 0.0983 & 0.0941 & 0.0822 & & 0.0637 & 0.0605 & 0.0513 \\
 & Ideological-Heuristic  & 0.0983 & 0.0975 & 0.0983 & & 0.0637 & 0.0612 & 0.0587 \\
 & Bridge-Heuristic       & 0.0983 & 0.0961 & 0.0954 & & 0.0637 & 0.0617 & 0.0591 \\
 & CLEAR-PGD ($q$-only)   & 0.0983 & 0.0952 & 0.0943 & & 0.0637 & 0.0594 & 0.0581 \\
 & CLEAR-PGD ($p$-only)   & 0.0983 & 0.0875 & 0.0804 & & 0.0637 & 0.0553 & 0.0502 \\
 & CLEAR-PGD              & 0.0983 & \textbf{0.0820} & \textbf{0.0734} & & 0.0637 & \textbf{0.0514} & \textbf{0.0432} \\
\bottomrule
\end{tabular}%
}
\vspace{-6pt}
\end{table}

\begin{table}[h!]
\centering
\caption{Pearson correlations ($r$) of shifts $(\Delta p_i, \Delta q_i)$ with $e_i =|s_i - 0.5|$, $\tau_i = |\Delta p_i| + |\Delta q_i|$, PageRank $\pi_i$, activity level of users, and budget split allocated between $\mathbf{p}$ and $\mathbf{q}$, for $b/n = 0.4$.}
\label{tab:appendix_rq3}
\vspace{-6pt}
\resizebox{\columnwidth}{!}{%
\begin{tabular}{llccccccccccc}
\toprule
\textbf{Dataset} & \textbf{Scen.} & \multicolumn{3}{c}{\textbf{Pearson Corr. ($r$)}} & \multicolumn{2}{c}{\textbf{Broadcasters}} & \multicolumn{2}{c}{\textbf{Contributors}} & \multicolumn{2}{c}{\textbf{Observers}} & \multicolumn{2}{c}{\textbf{Budget Split}} \\
\cmidrule(lr){3-5} \cmidrule(lr){6-7} \cmidrule(lr){8-9} \cmidrule(lr){10-11} \cmidrule(lr){12-13}
 & & $r(\Delta p, e)$ & $r(\Delta q, e)$ & $r(\tau, \pi)$ & $\overline{\Delta p}$ & $\overline{\Delta q}$ & $\overline{\Delta p}$ & $\overline{\Delta q}$ & $\overline{\Delta p}$ & $\overline{\Delta q}$ & $\% b_{\mathbf{p}}$ & $\% b_{\mathbf{q}}$ \\
\midrule
\multirow{2}{*}{\textbf{Wiki}} 
 & A & $-0.59$ & $+0.51$ & $+0.29$ & $-0.02$ & $+0.26$ & $+0.23$ & $+0.29$ & $+0.20$ & $0.00$  & $64.6\%$ & $35.4\%$ \\
 & B & $-0.35$ & $+0.72$ & $+0.21$ & $+0.03$ & $+0.27$ & $+0.26$ & $+0.16$ & $+0.24$ & $-0.03$ & $63.9\%$ & $36.1\%$ \\
\midrule
\multirow{2}{*}{\textbf{Gnutella}} 
 & A & $-0.32$ & $+0.35$ & $+0.32$ & $+0.08$ & $+0.26$ & $+0.26$ & $+0.30$ & $+0.24$ & $+0.04$ & $69.9\%$ & $30.1\%$ \\
 & B & $-0.35$ & $+0.68$ & $+0.21$ & $-0.02$ & $+0.18$ & $+0.25$ & $+0.17$ & $+0.22$ & $-0.02$ & $63.6\%$ & $36.4\%$ \\
\midrule
\multirow{2}{*}{\textbf{Epinions}} 
 & A & $-0.27$ & $+0.49$ & $+0.15$ & $+0.08$ & $+0.16$ & $+0.33$ & $+0.07$ & $+0.23$ & $+0.01$ & $65.1\%$ & $34.9\%$ \\
 & B & $-0.33$ & $+0.69$ & $+0.15$ & $+0.06$ & $+0.28$ & $+0.31$ & $+0.19$ & $+0.21$ & $-0.03$ & $61.4\%$ & $38.6\%$ \\
\midrule
\multirow{2}{*}{\textbf{Slashdot}} 
 & A & $-0.27$ & $+0.47$ & $+0.23$ & $+0.07$ & $+0.29$ & $+0.37$ & $+0.16$ & $+0.22$ & $+0.02$ & $65.7\%$ & $34.3\%$ \\
 & B & $-0.30$ & $+0.67$ & $+0.20$ & $+0.06$ & $+0.31$ & $+0.34$ & $+0.19$ & $+0.20$ & $-0.03$ & $61.2\%$ & $38.8\%$ \\
\midrule
\multirow{2}{*}{\textbf{YouTube}} 
 & A & $-0.24$ & $+0.42$ & $+0.12$ & $+0.05$ & $+0.25$ & $+0.33$ & $+0.28$ & $+0.21$ & $+0.01$ & $64.5\%$ & $35.5\%$ \\
 & B & $-0.22$ & $+0.60$ & $+0.08$ & $+0.07$ & $+0.50$ & $+0.37$ & $+0.25$ & $+0.19$ & $-0.02$ & $60.8\%$ & $39.2\%$ \\
\midrule
\multirow{2}{*}{\textbf{Pokec}} 
 & A & $-0.30$ & $+0.45$ & $+0.48$ & $+0.03$ & $+0.26$ & $+0.35$ & $+0.30$ & $+0.23$ & $+0.02$ & $72.0\%$ & $28.0\%$ \\
 & B & $-0.35$ & $+0.65$ & $+0.52$ & $+0.05$ & $+0.38$ & $+0.38$ & $+0.25$ & $+0.25$ & $+0.01$ & $68.6\%$ & $31.4\%$ \\
\bottomrule
\end{tabular}%
}
\vspace{-6pt}
\end{table}

\spara{Initial Setup:}
The results presented in RQ2 and RQ3 depend on the initial $\mathbf{p}^{(0)}$, $\mathbf{q}^{(0)}$, $\mathbf{s}$, and $\mathbf{a}$. 
We chose to apply the \textsc{CLEAR-PGD} and the baseline algorithms under this regime to evaluate the model on realistic network behaviors~\cite{nielsen2006participation,CarronArthur2014}. 
Our evaluation focuses on these  configurations and the underlying mathematical framework remains applicable to any other domain-specific initializations.

\end{document}